\documentclass[11pt]{article} 
\pdfoutput=1
\usepackage[utf8]{inputenc}
\usepackage[T1]{fontenc}
\usepackage[english]{babel}

\usepackage{amsmath,amsfonts,amssymb,amsthm}
\usepackage{enumitem}
\usepackage{mathtools}
\usepackage{booktabs}

\usepackage[mono=false]{libertine}
\usepackage[cmintegrals,libertine]{newtxmath}
\usepackage[cal=euler, scr=boondoxo]{mathalfa}

\useosf
\usepackage[a4paper,vmargin={3.5cm,3.5cm},hmargin={2.5cm,2.5cm}]{geometry}
\usepackage[font={small,sf}, labelfont={sf,bf}, margin=1cm]{caption}
\usepackage[colorlinks=true]{hyperref}
\usepackage{color,graphicx}

\usepackage{stackrel}
\usepackage{soul}

\theoremstyle{plain}
\newtheorem{theorem}{Theorem}
\newtheorem{corollary}[theorem]{Corollary}
\newtheorem{proposition}[theorem]{Proposition}
\newtheorem{lemma}[theorem]{Lemma}
\theoremstyle{definition}

\newcommand{\ind}{\mathbf{1}}
\newcommand{\R}{\mathbb{R}}

\newcommand{\Z}{\mathbb{Z}}
\newcommand{\Q}{\mathbb{Q}}
\newcommand{\rmd}{\mathrm{d}}
\newcommand{\map}{\mathfrak{m}}

\begin{document}
	
	\title{\bf Irreducible metric maps and Weil--Petersson volumes}
	\author{\textsc{Timothy Budd}\footnote{Radboud University, Nijmegen, The Netherlands. Email: \href{mailto:t.budd@science.ru.nl}{t.budd@science.ru.nl}}}
	\date{\today}
	\maketitle
	\begin{abstract}
		We consider maps on a surface of genus $g$ with all vertices of degree at least three and positive real lengths assigned to the edges.
		In particular, we study the family of such metric maps with fixed genus $g$ and fixed number $n$ of faces with circumferences $\alpha_1,\ldots,\alpha_n$ and a $\beta$-irreducibility constraint, which roughly requires that all contractible cycles have length at least $\beta$.
		Using recent results on the enumeration of discrete maps with an irreducibility constraint, we compute the volume $V_{g,n}^{(\beta)}(\alpha_1,\ldots,\alpha_n)$ of this family of maps that arises naturally from the Lebesgue measure on the edge lengths.
		It is shown to be a homogeneous polynomial in $\beta, \alpha_1,\ldots, \alpha_n$ of degree $6g-6+2n$ and to satisfy string and dilaton equations.
		Surprisingly, for $g=0,1$ and $\beta=2\pi$ the volume $V_{g,n}^{(2\pi)}$ is identical, up to powers of two, to the Weil--Petersson volume $V_{g,n}^{\mathrm{WP}}$ of hyperbolic surfaces of genus $g$ and $n$ geodesic boundary components of length $L_i = \sqrt{\alpha_i^2 - 4\pi^2}$, $i=1,\ldots,n$.
		For genus $g\geq 2$ the identity between the volumes fails, but we provide explicit generating functions for both types of volumes, demonstrating that they are closely related.
		Finally we discuss the possibility of bijective interpretations via hyperbolic polyhedra. 
	\end{abstract}
	\vspace{1cm}

	\begin{figure}[h]
		\centering
		\includegraphics[width=.95\linewidth]{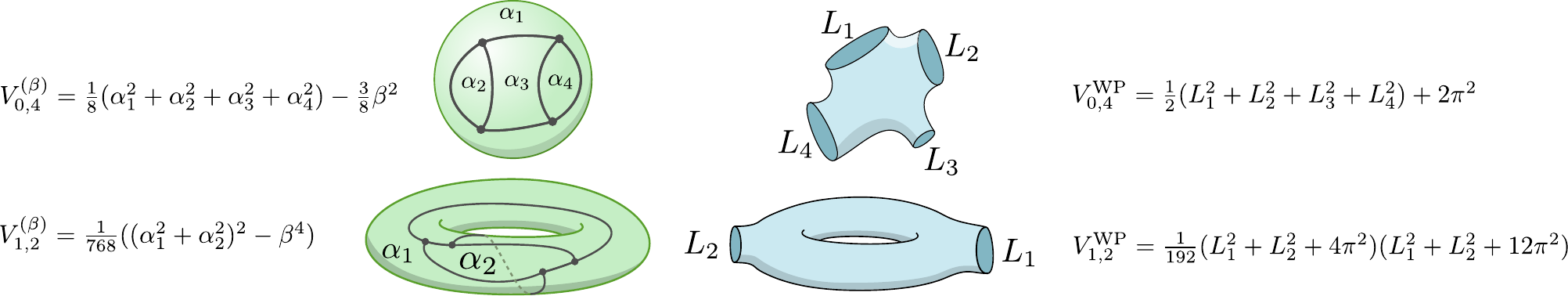}
		\caption{Examples of the volume polynomials $V_{g,n}^{(\beta)}(\alpha_1,\ldots,\alpha_n)$ of essentially $\beta$-irreducible metric maps of genus $g$ with $n$ labeled faces of circumferences $\alpha_1,\ldots,\alpha_n$, compared to the Weil--Petersson volumes $V_{g,n}^{\mathrm{WP}}(L_1,\ldots,L_n)$ of hyperbolic surfaces of genus $g$ with $n$ labeled geodesic boundaries of lengths $L_1,\ldots,L_n$.
		In the case $g=0,1$ and $\beta = 2\pi$ they satisfy $V_{g,n}^{(2\pi)} = 2^{2-2g-n} V_{g,n}^{\mathrm{WP}}$ under the identification $L_i^2 = \alpha_i^2 - 4\pi^2$ (see Corollary \ref{thm:metricwpequivalence}).
		\label{fig:volumecomp}}
	\end{figure}
	\clearpage

\section{Introduction}

How large is the space of metric structures one can put on a surface of given topology?
As stated it is not a well-posed question, but it can be turned into one by a suitable restriction on the class of metrics and/or choice of measure.
This can be done in various ways.
One can take the discrete approach, in which one counts the number of combinatorial maps (or ribbon graphs) one can draw on a surface with, say, fixed number of faces of specified degrees.
Roughly this can be understood as restricting the metric spaces to those of fixed topology that can be assembled from gluing a given collection of regular polygons of unit side length.
Another approach is to reduce the space of metric structures by imposing a constant curvature restriction.
It is classical that such surfaces are equivalently characterized by their complex structure, so one is really dealing with moduli spaces of Riemann surfaces.
These moduli spaces are naturally equipped with the Weil--Petersson symplectic structure that gives rise to a non-degenerate volume form, so one may quantify the size of the space of hyperbolic metric structures by computing the Weil--Petersson volumes of moduli spaces.
Both approaches have a long history in physics (e.g.\ two-dimensional quantum gravity, string theory) and mathematics (e.g.\ combinatorics, algebraic geometry, geometric topology) and share many features.

For instance, both are integrable in the sense that appropriately defined generating functions satisfy infinite hierarchies of partial differential equations (see e.g. \cite{Lando_Graphs_2004} and references).
Weil--Petersson volumes appear in solutions to the KdV hierarchy
\cite{Witten_Two_1991,Kontsevich1992,Manin_Invertible_2000,Mirzakhani2007a,Mulase_Mirzakhanis_2008,Liu_Recursion_2009},
while map enumeration problems are related to solutions of the KP and 2-Toda hierarchies \cite{Okounkov_Toda_2000,Goulden_KP_2008,Louf_new_2019}.
Both satisfy topological recursion equations (see \cite{Eynard2016} for an overview) in the form of Tutte's equation for maps \cite{Tutte_census_1963,Ambjorn_Matrix_1993,Eynard_Formal_2011} and Mirzakhani's recursion for Weil--Petersson volumes \cite{Mirzakhani2007}.
These are succinctly summarized in the data of a spectral curve \cite{Eynard_Invariants_2007,Eynard_Weil_2007}.
Also at the level of statistical properties of random surfaces there is some overlap, though the bulk of work in random maps has focussed on the planar case \cite{LeGall_Brownian_2019,LeGall_Scaling_2012} while random hyperbolic surfaces sampled proportionally to the Weil--Petersson volume have mainly been investigated in the large-genus regime \cite{Mirzakhani_Growth_2013}.
For recent progress on properties of random maps surfaces of large genus see \cite{Budzinski_Local_2020,Louf_Planarity_2020}. 

Despite the similarities there are still a lot of powerful combinatorial and probabilistic techniques that have not yet found their way from one approach to the other.
To facilitate the transfer of methods it would be of significant help to understand the relation between maps and hyperbolic surfaces at a bijective level, something that is currently lacking for general hyperbolic surfaces.
This work can be viewed as a first step in this direction, by identifying a class of maps for which in certain situations the associated volumes agree with the Weil--Petersson volumes.

Let us summarize how this rather unexpected coincidence arises from the combinatorics of maps when we disallow short cycles.  
In the planar case the general problem of enumerating maps with prescribed degrees and girth, i.e.\ the length of the shortest cycle, was made accessible by Bernardi \& Fusy \cite{Bernardi_bijection_2012,Bernardi_Unified_2012} who introduced a tree encoding for such maps. 
Planar maps with the slightly stronger constraint of $d$-irreducibility, which requires that the girth is at least $d$ and that the faces of degree $d$ are the only cycles of length $d$, were enumerated by Bouttier \& Guitter \cite{Bouttier2014,Bouttier2014a}.
In \cite{Budd2020} we used their methods to show that for each $n\geq 3$ and $b\geq 1$ the number of $2b$-irreducible planar maps with $n$ faces of prescribed evens degrees and with no vertices of degree one admits a simple formula, namely a polynomial in $b$ and the face degrees. 
This can be extended to higher genus if one imposes the irreducibility constraint on the universal cover of the maps, leading to the notion of \emph{essentially $2b$-irreducible} maps on surfaces of arbitrary genus.
Having the enumeration encoded in polynomials makes it very easy to study the limit in which both $b$ and the face degrees become large, meaning that we are counting maps in which all cycles are required to be long.
In this limit one is naturally led to consider continuous analogues of discrete maps in the form of \emph{metric maps}, which are maps with real lengths assigned to the edges.
Such maps admit an analogous criterion of irreducibility in which the discrete length $2b$ is replaced by a real minimal length $\beta$.
Since there is a continuum of (essentially) $\beta$-irreducible maps the question of enumeration becomes one of the computation of a volume with respect to a natural measure, just like in the case of hyperbolic surfaces.
In this case the measure arises in a simple fashion from the Lebesgue measure on the real length assignments to the edges.
As we will see these volumes are closely related to the Weil--Petersson volumes of hyperbolic surfaces with geodesic boundary components of prescribed lengths, whose general computation was first solved by Mirzakhani \cite{Mirzakhani2007}. 

Whether this relation has a natural bijective interpretation is still very much an open question in the general case.
There are two regimes in which a bijective interpretation is known.
The first of these is the limit in which the boundary lengths of the hyperbolic surfaces become large, in which case the underlying geometries literally approach (in a Gromov--Hausdorff sense) those of  metric maps \cite{Do_asymptotic_2010,Andersen_Kontsevich_2020} without irreducibility constraint.
These metric maps, obtained from the $\beta$-irreducible metric maps by taking $\beta$ to zero, are precisely the ones appearing in Kontsevich's proof \cite{Kontsevich1992} of Witten's conjecture.
The other regime corresponds to taking all boundary lengths to zero, resulting in hyperbolic surfaces with cusps.
In the planar case such surfaces can be related to $2\pi$-irreducible metric maps with faces of circumference $2\pi$ via two bijections of Rivin \cite{Rivin1992,Rivin1996} involving ideal polyhedra in three-dimensional hyperbolic space.
Computations by Charbonnier, David \& Eynard in \cite{David_Planar_2014,Charbonnier2017} show that this relation indeed identifies the corresponding metric map volumes with the Weil--Petersson volumes.
This connection will be summarized in Section~\ref{sec:polyhedra} and prospects for generalizations to different topologies and non-zero boundary lengths will be discussed.
However, we start by giving precise definitions and statements of our main enumerative results.

\subsection{Definition of metric map volumes}\label{sec:metricmapdef}
A \emph{genus-$g$ map} is a (multi)graph that is properly embedded in a surface of genus $g$, viewed up to orientation-preserving homeomorphisms of the surface.
Here \emph{properly embedded} means that edges only meet at their endpoints and that the complement of the graph is a disjoint union of topological disks.
We denote the set of vertices, edges, and faces of a map $\map$ by $\mathcal{V}(\map)$, $\mathcal{E}(\map)$ and $\mathcal{F}(\map)$ respectively.
A \emph{cubic} map is a map with all vertices of degree three.
A map is \emph{rooted} if it is equipped with a distinguished oriented edge, the \emph{root edge}.

A \emph{genus-$g$ metric map} is a genus-$g$ map with all vertices of degree at least three and a positive real number associated to each edge, which we interpret as the length of the edge.
A planar metric map is said to be \emph{$\beta$-irreducible} for some $\beta \geq 0$ if there is no simple cycle in the map that has length smaller than $\beta$ and each cycle of length $\beta$ corresponds to the contour of a face (of circumference $\beta$).
For $g\geq 1$, a genus-$g$ metric map is \emph{essentially $\beta$-irreducible} if its universal cover, seen as an infinite planar metric map, is $\beta$-irreducible.
We denote the set of such maps with $n$ labeled faces by $\mathcal{R}^{(\beta)}_{g,n}$, and the corresponding set of rooted maps by $\vec{\mathcal{R}}^{(\beta)}_{g,n}$.

The main quantity of interest is the volume $V_{g,n}^{(\beta)}(\alpha_1,\ldots,\alpha_n)$ of the maps in $\vec{\mathcal{R}}^{(\beta)}_{g,n}$ whose $i$th face has circumference exactly $\alpha_i$, where  the volume measure arises roughly from the Lebesgue measure on the edge lengths.
To avoid ambiguities in the normalization of the volume measure and peculiarities that can occur for special choices of $\alpha_1,\ldots,\alpha_n$, we take some care in its definition.
The set $\vec{\mathcal{R}}^{(\beta)}_{g,n}$ is naturally partitioned into subsets sharing the same underlying rooted map (forgetting the edge lengths).
A special role is played by the rooted maps that are \emph{cubic}, i.e.\ each vertex has degree three, which have a maximal number of edges equal to $6g-6+3n$.
We denote the set of these rooted genus-$g$ cubic maps by $\vec{\mathcal{C}}_{g,n}$.
For $\mathfrak{s}\in\vec{\mathcal{C}}_{g,n}$ the corresponding subset of $\vec{\mathcal{R}}^{(\beta)}_{g,n}$ is described by a convex open polytope in the space $\R_{>0}^{6g-6+3n}$ of possible length assignments $\theta_1,\ldots,\theta_{6g-g+3n}$ to its edges (see Section~\ref{sec:discrete} for details).
Let $\mu_{\mathfrak{s}}$ be the push-forward of the $(6g-6+3n)$-dimensional Lebesgue measure $\rmd\theta_1\cdots\rmd\theta_{6g-6+3n}$ on this polytope to $\vec{\mathcal{R}}^{(\beta)}_{g,n}$ and define the measure
\begin{equation}\label{eq:measuremu}
	\mu_{g,n} = \sum_{\mathfrak{s}\in\vec{\mathcal{C}}_{g,n}} \frac{\mu_{\mathfrak{s}}}{2|\mathcal{E}(\mathfrak{s})|}
\end{equation}
on $\vec{\mathcal{R}}^{(\beta)}_{g,n}$, which thus assigns zero measure to the non-cubic maps.
Let $\mathsf{Circ} : \vec{\mathcal{R}}^{(\beta)}_{g,n} \to [\beta,\infty)^{n}$ denote the assignment of circumferences to the $n$ labeled faces.
Then the push forward of $\mu_{g,n}$ along $\mathsf{Circ}$ defines a measure on $[\beta,\infty)^{n}$ that has a density with respect to the Lebesgue measure $\rmd\alpha_1\cdots\rmd\alpha_n$ on $[\beta,\infty)^{n}$.
We define the \emph{volume of essentially $\beta$-irreducible genus-$g$ metric maps with $n$ faces of circumferences $\alpha_1,\ldots,\alpha_n$} to be
\begin{align*}
V_{g,n}^{(\beta)}(\alpha_1,\ldots,\alpha_n) &=   \frac{\mathsf{Circ}_*\mu_{g,n}}{\rmd\alpha_1\cdots\rmd\alpha_n},
\end{align*}
where the second fraction represents the Radon--Nikodym derivatives. 
See Figure \ref{fig:maps3faces} and Figure \ref{fig:torusmap} for the simple examples of $(g,n)=(0,3)$ and $(g,n)=(1,1)$, illustrating that 
\begin{equation}\label{eq:initialvolumes}
	V_{0,3}^{(\beta)}(\alpha_1,\alpha_2,\alpha_3) = \frac{1}{2}, \qquad V_{1,1}^{(\beta)}(\alpha_1) = \frac{\alpha_1^2}{96}.
\end{equation}
These examples are special in the sense that the irreducibility constraint is vacuous. 
In general $V_{g,n}^{(\beta)}$ will of course depend on $\beta$.

\begin{figure}
	\centering
	\includegraphics[width=.95\linewidth]{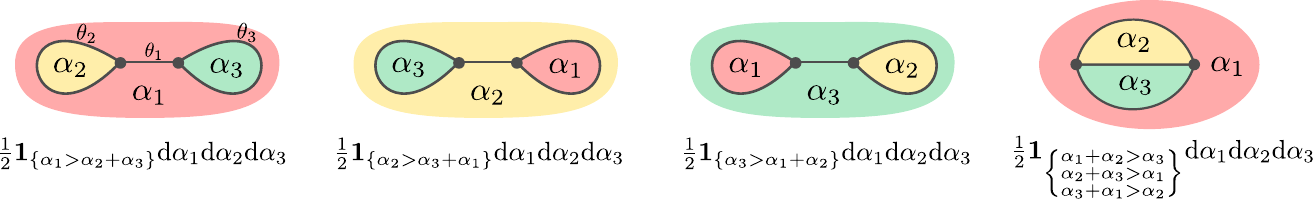}
	\caption{There are four cubic planar maps with three labeled faces (which can be rooted in six ways each). 
	The formula below each map $\mathfrak{s}$ gives the measure $\mathsf{Circ}_*\mu_{\mathfrak{s}}$ on $[\beta,\infty)^n$.
	For example, the measure $\tfrac{1}{2} \mathbf {1}_{\{\alpha_1>\alpha_2+\alpha_3\}}\mathrm{d}\alpha_1\mathrm{d}\alpha_2\mathrm{d}\alpha_3$ for the left map is obtained as the push-forward of the Lebesgue measure $\rmd\theta_1\rmd\theta_2\rmd\theta_3$ along $(\theta_1,\theta_2,\theta_3) \mapsto (2\theta_1+\theta_2+\theta_3,\theta_2,\theta_3)$.
	Adding up all four contributions gives $\mu_{0,3} = \tfrac{1}{2} \mathrm{d}\alpha_1\mathrm{d}\alpha_2\mathrm{d}\alpha_3$, hence $V_{0,3}^{(\beta)}(\alpha_1,\alpha_2,\alpha_3) = \tfrac12$.
	\label{fig:maps3faces}}
\end{figure}
\begin{figure}
	\centering
	\includegraphics[width=.2\linewidth]{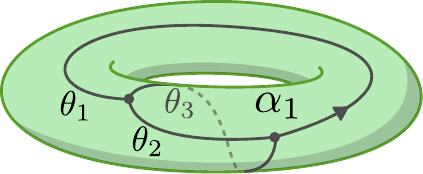}
	\caption{There exists a single rooted cubic genus-$1$ map $\mathfrak{s} \in \vec{\mathcal{C}}_{1,1}$.
	The push-forward of the Lebesgue measure $\rmd\theta_1\rmd\theta_2\rmd\theta_3$ along $(\theta_1,\theta_2,\theta_3) \mapsto \alpha_1=2\theta_1+2\theta_2+2\theta_3$ is $\frac{\alpha_1^2}{16}\rmd\alpha_1$. Hence $\mu_{1,1} = \frac{\alpha_1^2}{96} \rmd\alpha_1$ and $V_{1,1}^{(\beta)}(\alpha_1) = \frac{\alpha_1^2}{96}$.
	\label{fig:torusmap}}
\end{figure}

As a remark, we mention that we could have replaced the property of essential $\beta$-irreducibility by that of having essential girth at least $\beta$, where the \emph{essential girth} of a metric map is the minimal metric length of a simple cycle on the universal cover of the map.
Note that any essentially $\beta$-irreducible metric map necessarily has essential girth at least $\beta$.
In fact, the set $\vec{\mathcal{R}}^{(\beta)}_{g,n}$ of essentially $\beta$-irreducible metric maps is an open subset of full Lebesgue measure in the set of metric maps with essential girth at least $\beta$.
Indeed, the only maps that are missing in the former are those metric maps that have a simple cycle of length $\beta$ in the universal cover that does not bound a face.
These maps comprise a subset of positive codimension and thus have zero Lebesgue measure.
It follows that $V_{g,n}^{(\beta)}(\alpha_1,\ldots,\alpha_n)$ also measures the volume of metric maps with essential girth at least $\beta$, while $-\frac{\partial}{\partial \beta}V_{g,n}^{(\beta)}(\alpha_1,\ldots,\alpha_n)$ with $\alpha_i > \beta$, $i=1,\ldots,n$, is the volume of metric maps with essential girth exactly equal to $\beta$.

\subsection{Main results}

For the determination of $V_{g,n}^{(\beta)}$ we rely on the previous work \cite{Budd2020}, in which we studied the enumeration of discrete maps with an irreducibility constraint.
In particular, we considered genus-$g$ maps with $n$ faces of even degrees $2\ell_1,\ldots,2\ell_n$ with no vertices of degree one and that are essentially $2b$-irreducible for a positive integer $b$.
This discrete notion of irreducibility is very similar to the one described above: it requires that all simple cycles in the universal cover have (discrete) length at least $2b$ and they are only allowed to have length exactly $2b$ if they bound a face of degree $2b$. 
According to \cite[Theorem 1]{Budd2020} these maps are enumerated by polynomials $\hat{N}_{g,n}^{(b)}(\ell_1,\ldots,\ell_n)$ of degree $6g-6+2n$ in $b,\ell_1,\ldots,\ell_n$.
See Section~\ref{sec:discrete} for details and precise statements.
In Proposition \ref{thm:polynomiallimit} we show that the set of essentially $\beta$-irreducible metric maps naturally arises in the limit of large $b$ and $\ell$ and that the volume $V_{g,n}^{(\beta)}$ can be extracted from the leading order of the polynomials $\hat{N}_{g,n}^{(b)}$ via
\begin{equation*}
V_{g,n}^{(\beta)}(\alpha_1,\ldots,\alpha_n) = \frac{1}{2}\,[t^{6g-6+2n}]\, 
\hat{N}_{g,n}^{(\tfrac12\beta t)}(\tfrac12t\alpha_1,\ldots,\tfrac12t\alpha_n),
\end{equation*}
where $[t^k]f(t)$ denotes the coefficient of $t^k$ in $f(t)$. 
This allows us to derive the first main result describing several properties about the volume polynomials, which are analogous to the discrete results of \cite[Theorem 1]{Budd2020}.

\begin{theorem}\label{thm:realmaps}
	For any $\beta \in[ 0,\infty)$, $g\geq 0$ and $n\geq 1$ (provided $n \geq 3$ if $g=0$), $V^{(\beta)}_{g,n}(\alpha_1,\ldots,\alpha_n)$ is a symmetric polynomial of degree $3g-3+n$ in $\alpha_1^2,\ldots,\alpha_n^2$ and a homogeneous polynomial of degree $3g-3+n$ in $\beta^2,\alpha^2_1,\ldots,\alpha^2_n$.
	The polynomials satisfy the ``string'' and ``dilaton'' equations
	\begin{align}
	V^{(\beta)}_{g,n+1}(\alpha_1,\ldots,\alpha_n,0) &= \frac{1}{2}\sum_{j=1}^n \int_{\beta}^{\alpha_j} \alpha_j V^{(\beta)}_{g,n}(\alpha_1,\ldots,\alpha_n)\,\rmd\alpha_j, \label{eq:string} \\
	\frac{\partial^2V^{(\beta)}_{g,n+1}}{\partial\alpha_{n+1}^2}(\alpha_1,\ldots,\alpha_n,0) &= \frac{1}{2}(2g-2+n) V^{(\beta)}_{g,n}(\alpha_1,\ldots,\alpha_n).\label{eq:dilaton}
	\end{align}
	In the case $g\geq 1$ and $n=1$, $V^{(\beta)}_{g,1}(\alpha_1)$ is independent of $\beta$ and is given by a rational multiple of $\alpha_1^{6g-4}$.		
\end{theorem}	

\noindent
The proof of this theorem is the subject of Section~\ref{sec:discrete}.
Several polynomials for small $g$ and $n$ are listed in Table \ref{tab:polynomials}.
By homogeneously scaling all edge lengths and $\beta$ it is easy to see that for $\beta > 0$,
\begin{equation}\label{eq:Vscaling}
V_{g,n}^{(\beta)}(\alpha_1,\ldots,\alpha_n) = \left(\tfrac{\beta}{2\pi}\right)^{6g-6+2n} V_{g,n}^{(2\pi)}(\tfrac{2\pi}{\beta}\alpha_1,\ldots,\tfrac{2\pi}{\beta}\alpha_n).
\end{equation}
The reason for choosing $\beta = 2\pi$ as a natural reference value will become clear soon.

\begin{table}[t]
	\begin{tabular*}{\textwidth}{@{}ccl@{}} \toprule
		$g$ & $n$ & $V_{g,n}^{(\beta)}(\alpha_1, \ldots, \alpha_n)$ \\ \midrule
		$0$ & $3$ & $
		\frac{1}{2}$ \\
		& $4$ & $\frac{1}{8} m_{(1)}-\frac{3}{8} \beta ^2$ \\
		& $5$ & $\frac{1}{64} m_{(2)}+\frac{1}{16} m_{(1,1)}-\frac{3}{16} \beta ^2 \
		m_{(1)}+\frac{5}{16} \beta ^4$ \\
		& $6$ & $\frac{1}{768} m_{(3)}+\frac{3}{256} m_{(2,1)}+\frac{3}{64} \
		m_{(1,1,1)}-\frac{5}{128} \beta ^2 m_{(2)}-\frac{9}{64} \beta ^2 \
		m_{(1,1)}+\frac{15}{64} \beta ^4 m_{(1)}-\frac{215}{768} \beta ^6$\\ \midrule
		$1$ & $1$ & $\frac{1}{96} m_{(1)}$ \\
		& $2$ & $\frac{1}{768} m_{(2)}+\frac{1}{384} m_{(1,1)}-\frac{1}{768} \beta ^4 $\\
		& $3$ & $\frac{m_{(3)}}{9216}+\frac{m_{(2,1)}}{1536}+\frac{1}{768} m_{(1,1,1)}-\frac{\beta ^2 m_{(2)}}{3072}-\frac{\beta ^4 m_{(1)}}{1536}+\frac{\beta ^6 }{2304}$\\ 
		& $4$ & $\frac{m_{(4)}}{147456}+\frac{m_{(3,1)}}{12288}+\frac{m_{(2,2)}}{6144}+\frac{m_{(2,1,1)}}{2048}+\frac{m_{(1,1,1,1)}}{1024}-\frac{\beta ^2 m_{(3)}}{12288}-\frac{\beta ^2 m_{(2,1)}}{4096}$\\
		& & $\quad-\frac{\beta ^4 m_{(2)}}{12288}-\frac{\beta ^4 m_{(1,1)}}{2048}+\frac{\beta ^6 m_{(1)}}{3072}-\frac{\beta ^8 }{49152}$ \\ \midrule
		$2$ & $1$ & $\frac{m_{(4)}}{3538944}$ \\
		& $2$ & $\frac{m_{(5)}}{70778880}+\frac{m_{(4,1)}}{4718592}+\frac{29 m_{(3,2)}}{35389440}-\frac{\beta ^{10} }{70778880}$\\
		& $3$ & $\frac{m_{(6)}}{1698693120}+\frac{m_{(5,1)}}{70778880}+\frac{11 m_{(4,2)}}{141557760}+\frac{m_{(4,1,1)}}{4718592}+\frac{29 m_{(3,3)}}{212336640}+\frac{29 m_{(3,2,1)}}{35389440}+\frac{7 m_{(2,2,2)}}{3932160}$\\
		& & $\quad-\frac{\beta ^2 m_{(5)}}{283115520}-\frac{\beta ^4 m_{(4)}}{37748736}-\frac{29 \beta ^6 m_{(3)}}{424673280}-\frac{29 \beta ^8 m_{(2)}}{566231040}-\frac{\beta ^{10} m_{(1)}}{70778880}+\frac{\beta ^{12} }{169869312}$\\
		\midrule
		$3$ & $1$ & $\frac{m_{(7)}}{1712282664960}$ \\
		\bottomrule
	\end{tabular*}
	\caption{The first few polynomials $V_{g,n}^{(\beta)}(\alpha_1,\ldots,\alpha_n)$. They are expressed in terms of the basis $m_{(q_1,\ldots,q_k)}(\alpha_1,\ldots,\alpha_n)=\sum_{(p_1,\ldots,p_n)} \alpha_1^{2p_1}\alpha_2^{2p_2}\cdots\alpha_n^{2p_n}$ of symmetric even polynomials, where the sum runs over permutations $(p_1,\ldots,p_n)$ of $(q_1,\ldots,q_k,0,\ldots,0)$. For example, $m_{(2)}(\alpha_1,\alpha_2) = \alpha_1^4+\alpha_2^4$ and $m_{(3,1)}(\alpha_1,\alpha_2,\alpha_3) = \alpha_1^6\alpha_2^2+\alpha_1^6\alpha_3^2+\alpha_2^6\alpha_1^2+\alpha_2^6\alpha_3^2+\alpha_3^6\alpha_1^2+\alpha_3^6\alpha_2^2$.
	Analogous tables for $\hat{N}_{g,n}^{(b)}$ and $V_{g,n}^{\mathrm{WP}}$ can be found in \cite[Table 1]{Budd2020} and \cite[Appendix B]{Do2013} respectively.
	\label{tab:polynomials}}
\end{table}

The occurrence of polynomials in the volumes of metric maps is reminiscent of Weil--Petersson volumes and we will shortly see that they are closely related.
We start by recalling the definition of the Weil--Petersson volumes.
Let $\mathcal{M}_{g,n}(L_1,\ldots,L_n)$ be the moduli space of genus-$g$ hyperbolic surfaces with $n$ labeled geodesic boundary components of lengths $L_1,\ldots,L_n \geq 0$.
It comes with the Weil--Petersson symplectic structure $\omega$, which defines a natural volume measure $\omega^{3g-3+n} / (3g-3+n)!$ on $\mathcal{M}_{g,n}(L_1,\ldots,L_n)$. 
In a celebrated work \cite{Mirzakhani2007} Mirzakhani proved that the total Weil--Petersson volume $V^{\mathrm{WP}}_{g,n}(L_1,\ldots,L_n)$ of $\mathcal{M}_{g,n}(L_1,\ldots,L_n)$ is a symmetric polynomial in $L_1^2,\ldots,L_n^2$ of degree $3g-3+n$. 
Moreover, she obtained a recursion formula that determines all polynomials starting from $V_{0,3}^{\mathrm{WP}}(L_1,L_2,L_3) = 1$ and $V_{1,1}^{\mathrm{WP}}(L_1) = \frac{1}{48}(L_1^2+4\pi^2)$.
See \cite[Appendix B]{Do2013} for a table of Weil--Petersson volumes for small $g$ and $n$ in a format similar to Table \ref{tab:polynomials}.
The similarities should be apparent, especially when looking at the coefficients of the basis polynomials of top degree.

In \cite{Do2009} Do and Norbury proved that the Weil--Petersson volumes satisfy the string and dilaton equations
\begin{align}
	V^{\mathrm{WP}}_{g,n+1}(L_1,\ldots,L_n,2\pi i) &= \sum_{j=1}^n \int_{0}^{L_j} L_j V^{\mathrm{WP}}_{g,n}(L_1,\ldots,L_n)\,\rmd L_j, \label{eq:stringwp} \\
\frac{\partial V^{\mathrm{WP}}_{g,n+1}}{\partial L_{n+1}}(L_1,\ldots,L_n,2\pi i) &= 2\pi i(2g-2+n) V^{\mathrm{WP}}_{g,n}(L_1,\ldots,L_n).\label{eq:dilatonwp}
\end{align}
It is easily seen that these equations are identical to \eqref{eq:string} and \eqref{eq:dilaton}, up to powers of two, after the substitution $L_i = \sqrt{\alpha_i^2 - \beta^2}$ and using the fact that they are even polynomials.
As was shown in \cite[Theorem 4]{Do2009}, in the case $g=0$ and $g=1$ the string and dilaton equations uniquely determine the symmetric polynomials for arbitrary $n$ in terms of the base cases $V_{0,3}^{\mathrm{WP}} = 1$ and $V_{1,1}^{\mathrm{WP}}=\frac{1}{48}(L_1^2+4\pi^2)$.
Comparing to \eqref{eq:initialvolumes}, we immediately deduce the following identity (stated for $\beta=2\pi$, but the general relation is easily deduced from \eqref{eq:Vscaling}).

\begin{corollary}\label{thm:metricwpequivalence}
	For $g=0$ and $n\geq 3$, or $g=1$ and $n\geq 1$, we have the identity 
	\begin{equation}\label{eq:metricmapvsWP}
	V^{(2\pi)}_{g,n}(\alpha_1,\ldots,\alpha_n) = 2^{2-2g-n}\, V_{g,n}^{\mathrm{WP}}\left(\sqrt{\alpha_1^2-4\pi^2},\ldots,\sqrt{\alpha_1^2-4\pi^2}\right).
	\end{equation}  
\end{corollary}

\noindent 
For higher genus the identity \eqref{eq:metricmapvsWP} certainly does not hold. 
In particular, for $g\geq 2$ and $n=1$ we know from Theorem~\ref{thm:realmaps} that $V^{(2\pi)}_{g,1}(\alpha_1) = V^{(0)}_{g,1}(\alpha_1)$ is a monomial in $\alpha_1$ of degree $6g-6+2n$, while $V^{\mathrm{WP}}_{g,1}(\sqrt{\alpha_1^2 - 4\pi^2})$ is certainly not.
Nevertheless one may find a close connection between the two volumes at the level of generating functions.
To this end we fix an integer $d\geq 0$ and real numbers $\alpha_1, \ldots,\alpha_d \geq \alpha_0=2\pi$ and make the implicit identification $L_i = \sqrt{\alpha_i^2 - 4\pi^2}$, such that in particular $L_0=0$, and introduce the formal power series
\begin{align*}
F^{(2\pi)}_g(x_0,\ldots,x_d) &= \sum_{n=1}^\infty\frac{1}{n!} \sum_{i_1=0}^d x_{i_1} \cdots\! \sum_{i_n=0}^d x_{i_n} V^{(2\pi)}_{g,n}(\alpha_{i_1}\ldots,\alpha_{i_n}),\\
F_g^{\mathrm{WP}}(x_0,\ldots,x_d) &=  \sum_{n=0}^\infty\frac{1}{n!} \sum_{i_1=0}^d x_{i_1} \cdots\! \sum_{i_n=0}^d x_{i_n} 2^{2-2g-n}\,V^{\mathrm{WP}}_{g,n}(L_{i_1},\ldots,L_{i_n}).
\end{align*}
By convention $V_{0,0}^{\mathrm{WP}} = V_{0,1}^{\mathrm{WP}} =V_{0,2}^{\mathrm{WP}} = V_{1,0}^{\mathrm{WP}}=0$ because they fall outside of the stable range. Take note of the factor $2^{2-2g-n}$ in the second generating function!
By Corollary \ref{thm:metricwpequivalence} these generating functions are identical for $g\leq 1$.

In order to state expressions for these generating functions, we need to introduce several formal power series.
A central role is played by the formal power series $R(x_0,\ldots,x_d) \equiv \sum_{i=0}^d x_i + \cdots$ defined to be the unique solution to
\begin{equation}\label{eq:Rdef}
	Z(R) = 0,\qquad Z(r)\coloneqq \frac{\sqrt{r}}{\pi}J_1(2\pi\sqrt{r}) - \sum_{i=0}^d x_i \,I_0(L_i\sqrt{r})
\end{equation}
where $I_0$ and $J_1$ are (modified) Bessel functions.
Let also $M^{\mathrm{WP}}_k(x_0,\ldots,x_d)$ and $M^{(2\pi)}_k(x_0,\ldots,x_d)$ be the power series defined recursively via
\begin{equation}\label{eq:Mrecurrence}
	M^{\mathrm{WP}}_0 = M^{(2\pi)}_0 = \frac{1}{\partial_{x_0}R},\quad M^{\mathrm{WP}}_{k} = M^{\mathrm{WP}}_0 \,\partial_{x_0} M^{\mathrm{WP}}_{k-1}, \quad M_{k}^{(2\pi)} = M_0^{(2\pi)}\left(\tfrac{\pi^{2k}}{k!}+\partial_{x_0}M^{(2\pi)}_{k-1}\right),\quad k\geq 1.
\end{equation}
Finally, for $g\geq 2$ we introduce the polynomial
\begin{equation}\label{eq:Ppol}
	\mathcal{P}_g(m_1,\ldots,m_{3g-3}) = \sum_{\substack{d_2,d_3,\ldots\geq 0\\ \sum_{k\geq 2} (k-1)d_k = 3g-3}}\!\!\!\! \langle \tau_2^{d_2}\tau_3^{d_3}\cdots \rangle_g \prod_{k\geq 2} \frac{(-m_{k-1})^{d_k}}{d_k!},
\end{equation}
where $\langle \tau_2^{d_2}\tau_3^{d_3}\cdots \rangle_g$ are the $\psi$-class intersection numbers on the moduli space $\mathcal{M}_{g,n}$ with $n = \sum_k d_k \leq 3g-3$ marked points (see Section~\ref{sec:intersection} for details).
For instance, for $g=2,3$ the polynomials read (see e.g. \cite[(5.27)-(5.30)]{Itzykson_Combinatorics_1992} where $m_k = - I_{k+1}/ (1-I_1)$)
\begin{align*}
	\mathcal{P}_2 &= -\frac{7 m_1^3}{1440}+\frac{29 m_2 m_1}{5760}-\frac{m_3}{1152},\\
	\mathcal{P}_3 &= \frac{245 m_1^6}{20736}-\frac{193 m_2 m_1^4}{6912}+\frac{53 m_3 m_1^3}{6912}+\frac{205 m_2^2 m_1^2}{13824}-\frac{17 m_4 m_1^2}{11520}-\frac{1121 m_2 m_3 m_1}{241920}\\
	&\quad +\frac{77 m_5 m_1}{414720}-\frac{583 m_2^3}{580608}+\frac{607 m_3^2}{2903040}+\frac{503 m_2 m_4}{1451520}-\frac{m_6}{82944}.
\end{align*}

\begin{theorem}[Generating functions for metric maps and WP volumes]\label{thm:genfun}
	The generating functions $F_g^{(2\pi)}$ of essentially $2\pi$-irreducible metric maps are given by
	\begin{align}
	F^{(2\pi)}_0 &= \frac{1}{4}\int_{0}^{R} \!\!\!\rmd r\, Z(r)^2,\label{eq:genfun0}\\
	F^{(2\pi)}_1 &= - \frac{1}{24} \log M^{(2\pi)}_0,\label{eq:genfun1}\\
	F^{(2\pi)}_g &= \frac{2^{g-1}}{\left(M^{(2\pi)}_0\right)^{2g-2}}\,\mathcal{P}_g\left(\frac{M^{(2\pi)}_1}{M^{(2\pi)}_0},\ldots, \frac{M^{(2\pi)}_{3g-3}}{M^{(2\pi)}_0}\right) \quad \text{for }g\geq 2.\label{eq:genfun2}
	\end{align}
	The generating functions $F^{\mathrm{WP}}_g$ of Weil--Petersson volumes are given by the same formulas but with $M^{(2\pi)}_k$ replaced by $M^{\mathrm{WP}}_k$.
\end{theorem}

\noindent
The proof of this theorem is split into two parts. 
In Section~\ref{sec:genfunirrmaps}, more precisely Proposition~\ref{thm:metricmapgf}, we prove the formulas for metric maps in the slightly more general essentially $\beta$-irreducible case, except that the polynomials $\mathcal{P}_g$ remain unidentified. 
In Section~\ref{sec:wpintersection} we prove the analogous formulas for the Weil--Petersson volumes and identify the polynomials $\mathcal{P}_g$ through relations with intersection numbers.
Formulas like these for the generating functions of Weil--Petersson volumes have appeared before in various forms of generality and various degrees of rigour \cite{Itzykson_Combinatorics_1992,Zograf_Weil_1998,Dubrovin_Normal_2001,Okuyama_JT_2020}, but for convenience we have included more or less self-contained proof (in Section~\ref{sec:proofgenfun}).
Let us make several remarks regarding the expressions. \\[1mm]
\textbf{One face/boundary. }As mentioned before $V^{(2\pi)}_{g,1}(\alpha_1)$ is monomial in $\alpha_1$ while $V^{(\mathrm{WP})}_{g,1}(\sqrt{\alpha_1^2-4\pi^2})$ is not for $g\geq 2$. 
	Let us see how this comes about from the point of view of the generating functions.
	The only difference is the term $\pi^{2k} / k!$ in the recurrence \eqref{eq:Mrecurrence}, which ensures that (see Lemma \ref{lem:Mexpansion})
	\begin{equation*}
		M^{(2\pi)}_k(0,\ldots,0) = \ind_{k=0}, \qquad\text{whereas} \qquad M^{\mathrm{WP}}_k(0,\ldots,0) = \frac{(-\pi^2)^k}{k!}.
	\end{equation*}
	The first consequence is that 
	\begin{equation*}
		F_g^{(2\pi)}\big|_{x_i=0} = 0, \qquad\text{while}\qquad F_g^{\mathrm{WP}} \big|_{x_i=0} = 2^{2-2g}\, V_{g,0}^{\mathrm{WP}} = 2^{g-2}  \mathcal{P}_g\left( \frac{-\pi^2}{1!}, \frac{\pi^4}{2!},\frac{-\pi^6}{3!},\ldots \right) \neq 0
	\end{equation*}
	as expected. 
	On the other hand, for a single face of perimeter $\alpha_1$ we get 
	\begin{equation*}
		V_{g,1}^{(2\pi)}(\alpha_1) = \partial_{x_1}F_g^{(2\pi)}\Big|_{x_i=0} = - 2^{g-1} \langle \tau_{3g-2}\rangle_g \,\partial_{x_1} M_{3g-3}^{(2\pi)}\Big|_{x_i=0} = \frac{2^{3-5g}}{(3g-2)!} \langle \tau_{3g-2}\rangle_g \,\alpha_1^{6g-4},
	\end{equation*}
	where in the last equality we used that Lemma \ref{lem:Mexpansion} implies $\partial_{x_1} M_{p}^{(2\pi)}\big|_{x_i=0} = -4^{-p-1} \alpha_1^{2p+2} / (p+1)!$.
	We see that it is indeed a monomial in the face perimeter $\alpha_1$.
	In case of the Weil--Petersson volume $V_{g,1}^{\mathrm{WP}}(L_1)$, however, all terms in the polynomial $\mathcal{P}_g$ contribute, so the answer is not as simple (in particular not monomial in $\alpha_1$).\\[1mm]
\textbf{Explicit formulae in genus $0$.} Let us list some more explicit formulas for genus $g=0$ when one or more faces (or boundaries in the Weil--Petersson case) are distinguished. 
Taking derivatives of \eqref{eq:genfun0} while noting that $Z(R)=0$ we easily find
\begin{align*}
	\frac{\partial}{\partial x_{1}}F^{\mathrm{WP}}_0 &= -\frac{1}{2}\int_0^R I_0(L_1\sqrt{r})Z(r)\rmd r, \\
	\frac{\partial}{\partial x_{1}}\frac{\partial}{\partial x_{2}}F^{\mathrm{WP}}_0 &= \frac{1}{2} \int_0^R I_0(L_1\sqrt{r})I_0(L_2\sqrt{r})\rmd r,\\
	\frac{\partial}{\partial x_{1}}\frac{\partial}{\partial x_{2}}\frac{\partial}{\partial x_{3}}F^{\mathrm{WP}}_0 &= \frac{1}{2} \,I_0(L_1\sqrt{R})I_0(L_2\sqrt{R})I_0(L_3\sqrt{R}) \partial_{x_0}R.
\end{align*}
In each case the integral can be performed explicitly using integral identities like \cite[Equation 5.11(8)]{Watson_Treatise_1995}. 
For instance, 
\begin{align*}
	\frac{\partial}{\partial x_{1}}\frac{\partial}{\partial x_{2}}F^{\mathrm{WP}}_0 &= \frac{L_1\sqrt{R}\, I_1(L_1\sqrt{R})\,I_0(L_2\sqrt{R})-L_2\sqrt{R} \,I_1(L_2\sqrt{R})\,I_0(L_1\sqrt{R})}{L_1^2-L_2^2},\\
	\frac{\partial}{\partial x_{1}}F^{\mathrm{WP}}_0 &= -4\frac{2 \pi  \sqrt{R} J_1\left(2 \pi  \sqrt{R}\right) I_0\left(L_1\sqrt{R} \right)+L_1 \sqrt{R} J_0\left(2 \pi  \sqrt{R}\right) I_1\left(L_1\sqrt{R}
	\right)}{\left(L_1^2+4 \pi ^2\right)^2}\\
	&\quad+\frac{2 \pi  R J_0\left(2 \pi 
	\sqrt{R}\right) I_0\left(L_1\sqrt{R} \right)-L_1 R J_1\left(2 \pi  \sqrt{R}\right) I_1\left(L_1\sqrt{R} \right)}{\pi(L_1^2+4 \pi^2)}  + \sum_{i=0}^d x_i\, \frac{\partial}{\partial x_{1}}\frac{\partial}{\partial x_{i}}F^{\mathrm{WP}}_0.
\end{align*}
In particular, if the second boundary is a cusp ($L_2=0$) we obtain
\begin{equation*}
	\frac{\partial}{\partial x_{0}}\frac{\partial}{\partial x_{1}}F^{\mathrm{WP}}_0 = \frac{\sqrt{R}}{L_1}I_1(L_1\sqrt{R}),
\end{equation*}
while if both are cusps then it is simply $\frac{\partial^2F^{\mathrm{WP}}_0 }{\partial x_{0}^2}=R/2$, providing a simple interpretation to the generating function $R$.

\subsection{Interpretation via hyperbolic polyhedra}\label{sec:polyhedra}

\textbf{Convex ideal polyhedra.\,}
Let us concentrate on the special case of essentially $2\pi$-irreducible metric maps with $n$ faces of circumference $\alpha_0 = 2\pi$.
In this case it is easy to see that every edge has length in $(0,\pi)$.
Indeed, if an edge $e$ would have length $\pi$ or larger, one could find a contractible cycle of length at most $2\pi$ by concatenating the contours (minus $e$ itself) of the two faces adjacent to $e$.

In the case of genus $g=0$ these metric maps are known since work of Rivin \cite{Rivin1996} to have an interpretation in hyperbolic geometry.
To understand this consider the three-dimensional hyperbolic space $\mathbb{H}^3$ in the Poincar\'e ball model, in which the open unit ball in $\R^3$ is equipped with a constant negative curvature metric such that geodesics correspond to circle segments that are orthogonal to the boundary.
An \emph{ideal polyhedron} is a non-compact polyhedron in $\mathbb{H}^3$ whose vertices all lie on the boundary sphere.
Convex ideal polyhedra are characterized by the positions of these vertices in the unit $2$-sphere, which are uniquely determined up to M\"obius transformation of the $2$-sphere. 
See Figure \ref{fig:surface}b for an example.

\begin{figure}
	\centering
	\includegraphics[width=\linewidth]{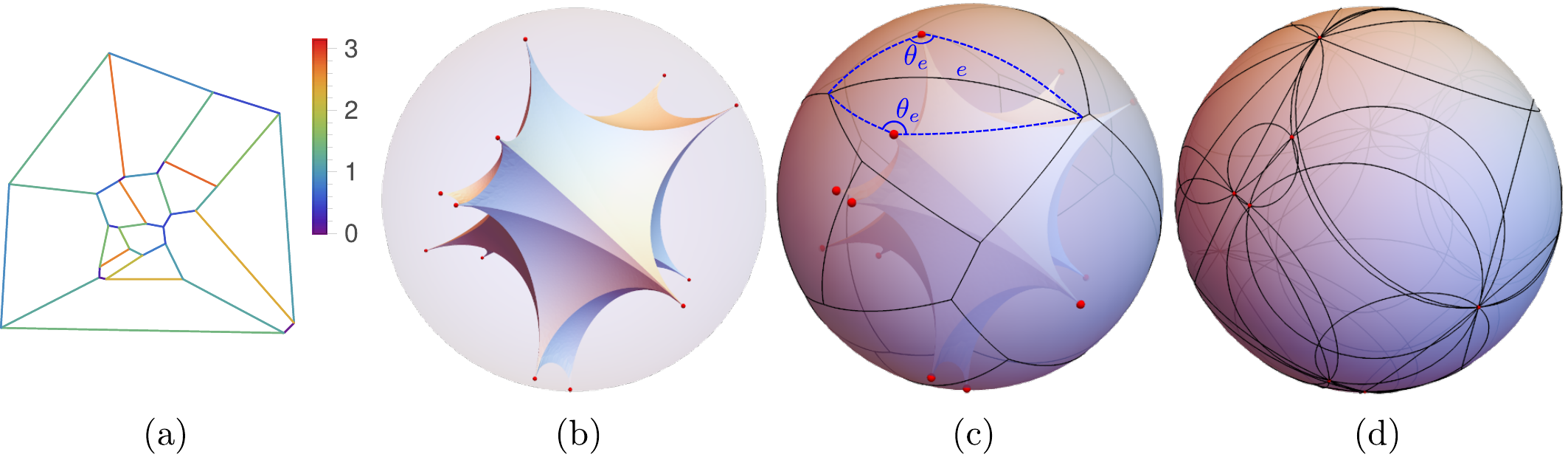}
	\caption{(a) a $2\pi$-irreducible metric map with all face of circumference $2\pi$ and edge lengths indicated by color shading; (b) the corresponding convex ideal polyhedron; (c) the Vorono\"i diagram of the vertices of the polyhedron, which can be turned into a metric map by assigning length $\theta_e$ to an edge $e$ where $\theta_e$ is the indicated angle as seen from a neighbouring vertex; (d) the corresponding circle pattern on the sphere.\label{fig:surface}}
\end{figure}

One may naturally associate a metric map to a convex ideal polyhedron, by considering the \emph{Vorono\"i diagram} of the $n$ vertices in the sphere (Figure \ref{fig:surface}c), i.e. the planar map whose underlying set consists of those points in the sphere that have more than one closest vertex among the $n$ (with respect to the standard round metric on the sphere).
An edge $e$ of the Vorono\"i diagram is given a length $\theta_e$ equal to the angle as seen from the vertex in the center of either one of its neigbouring Vorono\"i cells.
This angle is also precisely the external dihedral angle of the edge of the convex ideal polyhedron dual to $e$.
It is easy to see that the resulting metric map is invariant under M\"obius transformations of the vertex set in the sphere and that each face has circumference $2\pi$.
If one examines the angle sums of the spherical polygons corresponding to simple cycles in the map, one will note that the metric map is also $2\pi$-irreducible.
It is a non-trivial theorem of Rivin \cite{Rivin1996} that this determines a bijection between convex ideal polyhedra with $n$ vertices and $2\pi$-irreducible metric maps with $n$ faces of circumference $2\pi$.

To see that there is a connection with the moduli space of $\mathcal{M}_{0,n}(0,\ldots)$ of genus-$0$ hyperbolic surfaces with $n$ cusps (boundaries of length $0$), one need only look at Figure \ref{fig:surface}b and observe that the boundary faces of the polyhedron correspond to ideal polygons in a totally geodesic plane that carries the metric of the hyperbolic plane. 
Hence the boundary of a convex ideal polyhedron naturally has the structure of a genus-$0$ hyperbolic surface with $n$ cusps.
An earlier theorem of Rivin \cite{Rivin1992} shows that any element of $\mathcal{M}_{0,n}(0,\ldots)$ is uniquely obtained in this way.
The combination of both of Rivin's results therefore provides a bijection between $2\pi$-irreducible metric maps with $n$ faces of perimeter $2\pi$ and the moduli space $\mathcal{M}_{0,n}(0,\ldots)$.
A computation by Charbonnier, David \& Eynard in \cite{Charbonnier2017} (building on \cite{David_Planar_2014}) shows that the natural Lebesgue measure on metric maps is mapped via this bijection to the Weil--Petersson measure (up to a power of $2$), thus giving a bijective interpretation to the special case 
\begin{equation*}
	V_{0,n}^{(2\pi)}(2\pi,2\pi,\ldots) = 2^{2-n} V_{0,n}^{\mathrm{WP}}(0,0,\ldots)
\end{equation*}
of Corollary \ref{thm:metricwpequivalence}.

As an aside let us record here that even the enumeration of discrete irreducible maps has an interpretation in terms of compact ideal polyhedra.
The short proof based on \cite{Bouttier2014,Budd2020} is contained in the appendix.

\begin{proposition}\label{thm:discretepolyhedra}
	For $b\geq 2$ the exponential generating function for the number $P_{n}^{(b)}$ of ideal convex polyhedra with $n\geq 4$ labeled vertices and dihedral angles in $\frac{\pi}{b}\Z$ is given by 
	\begin{align*}
		\sum_{n=4}^\infty P_n^{(b)} \frac{z^{n-2}}{(n-2)!} &= \frac{1}{2(1+z)^2}- \frac{1}{2} + \frac{1}{2b}-\frac{1}{2b} \left(1+J^{-1}(b;z) \right)^{-2b}\\
		&= \frac{1}{2}(b-1)(b-2) z^2 + \frac{5}{12}(b+1)(b-1)(b-2)^2 z^3 + \cdots
	\end{align*}
	where $z\mapsto J^{-1}(b;z)$ is the formal power series inverse to $r\mapsto J(b;r) = r\,{_2F_1}(1-b,-b;2;-r)$ and ${_2F_1}$ is the hypergeometric function ${_2F_1}(a,b;c;z) = \sum_{n=0}^\infty \frac{(a)_n(b)_n}{(c)_n}\frac{z^n}{n!}$.
	In particular, $P_n^{(b)}$ is a polynomial in $b$ of degree $2n-6$ for any $n\geq 4$.
\end{proposition}

\begin{figure}[h]
	\centering
	\includegraphics[width=.6\linewidth]{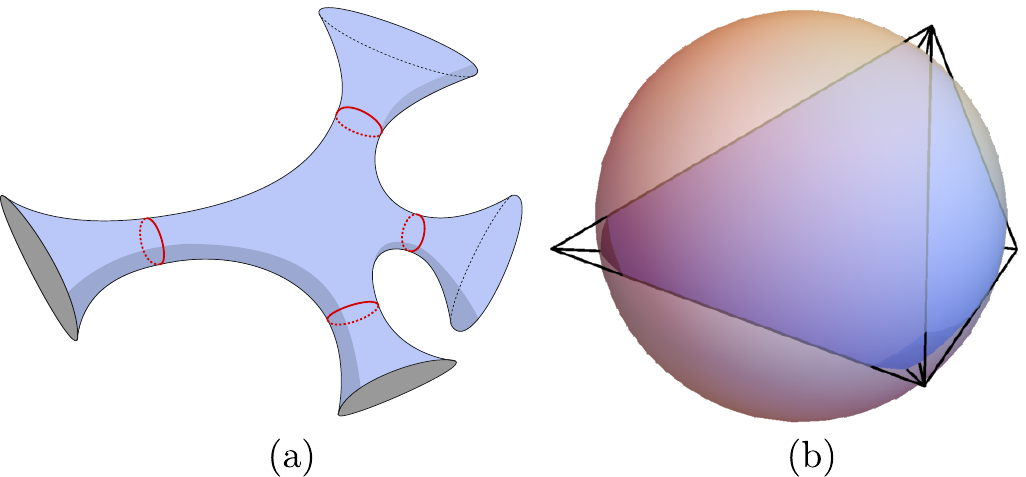}
	\caption{(a) A hyperbolic surface with $4$ geodesic boundaries extended to a complete hyperbolic surface; (b) a hyperideal convex polyhedron represented in the Klein model by a Euclidean polyhedron with all vertices outside the sphere and all edges meeting the interior of the sphere.\label{fig:hyperideal}}
\end{figure}

\noindent
\textbf{Hyperbolic surfaces with geodesic boundaries.}
It is natural to ask whether this bijective explanation extends to higher genus or faces of larger circumference.
Let us first discuss the case of genus-$0$ hyperbolic surfaces with non-zero boundary lengths.
Instead of hyperbolic surfaces with geodesic boundaries we may consider the complete hyperbolic surfaces obtained by gluing to each boundary component an infinite hyperbolic cylinder (Figure \ref{fig:hyperideal}a).
It was shown by Schlenker \cite{Schlenker_Metriques_1998} (see also \cite[Theorem 1.1]{Fillastre_Polyhedral_2008}) that any such surface can be uniquely realized as the boundary of a \emph{hyperideal convex polyhedron}, which is a polyhedron in 3-dimensional hyperbolic space whose vertices all lie beyond infinity. 
These polyhedra are most easily understood in the Klein model of three-dimensional hyperbolic space in the open unit ball in which geodesics are straight line segments. 
In this model a convex hyperideal polyhedron corresponds to the intersection of the open ball with a convex Euclidean polyhedron that has all  vertices strictly outside the unit sphere but such that all its edges meet the interior of the ball (Figure \ref{fig:hyperideal}b).
Just like in the ideal case, a convex hyperideal polyhedron with $n$ vertices is characterized by a metric map with $n$ faces with the edge lengths given by the external dihedral angles. 
Bao \& Bonahon \cite[Theorem 1]{Bao_Hyperideal_2002} identified the set of metric maps arising in this way (see also \cite{Springborn_variational_2008} for an argument based on a variational principle).
They are precisely the $2\pi$-irreducible metric maps with faces of arbitrary circumference, but with two additional requirements: all edge lengths should be shorter than $\pi$ (which is obviously necessary if they are to be external dihedral angles of a convex polyhedron) and any simple path starting and ending at a vertex on the same face but not contained in the contour of that face should have length larger than $\pi$.
Combining with the result of Schlenker this set of $2\pi$-irreducible metric maps is in bijection with hyperbolic surfaces with $n$ boundaries of arbitrary length.
Because this set of metric maps has additional restrictions beyond $2\pi$-irreducibility it is unlikely that this bijection can explain the general phenomenon of Corollary \ref{thm:metricwpequivalence} for $g=0$.
Moreover, the relation between the face circumference $\alpha_i$ and the boundary length $L_i$ will not be nearly as simple as $L_i=\sqrt{\alpha_i^2-4\pi^2}$ and it appears that the Weil--Petersson measure is not mapped to a measure as simple as the Lebesgue measure on the edge lengths.
We thus leave it as an open question to find a bijective interpretation of Corollary \ref{thm:metricwpequivalence}.\\[1mm]
\textbf{Higher genus.\,}
Let us finally consider the case of essentially $2\pi$-irreducible metric maps of arbitrary genus $g\geq 0$ and $n$ faces of circumference $2\pi$.
Known results are conveniently stated in terms of circle patterns.
A \emph{circle pattern} on a constant-curvature surface of genus $g$ is an embedded genus-$g$ map with all faces of degree at least $3$ and such that for each face the vertices lie on a circle.
See Figure \ref{fig:surface}d for an example in genus $0$ where only the circles and vertices are drawn.
A circle pattern naturally has a dual essentially $2\pi$-irreducible metric map with faces of circumference $2\pi$, in which the vertices correspond to the circles and the length of an edge is the (exterior) intersection angle between the neighbouring circles.
According to \cite[Theorem 4]{Bobenko_Variational_2004} this determines a bijection between essentially $2\pi$-irreducible metric map with faces of circumference $2\pi$ and circle patterns with $n$ vertices on genus-$g$ surfaces of constant curvature (equal to $1$ for $g=0$, $0$ for $g=1$, and $-1$ for $g\geq 2$).
The latter are viewed up to M\"obius transformations in the case $g=0$ and up to similarity transformations in the case $g=1$. 
When these circle patterns are equipped with a measure arising from the Lebesgue measure on the intersection angles of the circles, their total volumes is thus given by $V^{(2\pi)}_{g,n}(2\pi,2\pi,\ldots)$.

\begin{figure}
	\centering
	\includegraphics[width=.9\linewidth]{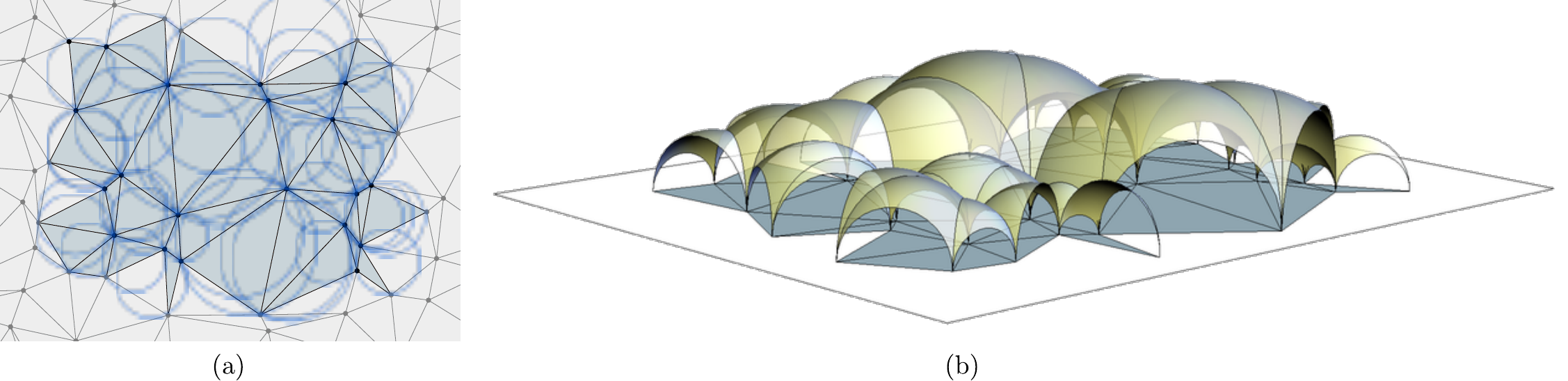}
	\caption{(a) A bi-periodic circle pattern in the plane. Only the circles corresponding to a fundamental domain of the torus are shown. (b) The corresponding boundary of the bi-periodic convex ideal polyhedron in the upper-half space model of $\mathbb{H}^3$. \label{fig:toruspattern}}
\end{figure}
In the case $g=1$, the circle pattern may be interpreted as an infinite biperiodic circle pattern on the Euclidean plane (Figure \ref{fig:toruspattern}a).
Viewing the Euclidean plane as the boundary of 3-dimensional hyperbolic space (in the Poincar\'e upper-half space model), we obtain a unique infinite biperiodic ideal convex polyhedron with vertices corresponding to the nodes of the circle pattern (Figure \ref{fig:toruspattern}b).
The boundary of this polyhedron is a $\Z^2$-covering of a unique genus-$1$ hyperbolic surface with $n$ cusps in $\mathcal{M}_{1,n}(0,0,\ldots)$.
It seems likely, but not proven, that this gives a bijective explanation for the special case
\begin{equation*}
	V_{1,n}^{(2\pi)}(2\pi,2\pi,\ldots) = 2^{-n} V_{1,n}^{(\mathrm{WP})}(0,0,\ldots)
\end{equation*}
of Corollary \ref{thm:metricwpequivalence}.
For genus $g\geq 2$, according to \cite[Lemma 4.23 \& Theorem 4.25]{Schlenker_Hyperbolic_2001} $2\pi$-irreducible metric maps with $n$ faces of circumference $2\pi$ describe the external dihedral angles of \emph{fuchsian ideal polyhedra} of genus $g$ (see \cite[Section 1.4]{Schlenker_Hyperbolic_2001} for definitions), whose boundaries correspond to coverings of genus-$g$ hyperbolic surfaces with $n$ cusps.
Here it is not quite clear what this means at the level of the measures.
Note that by Theorem~\ref{thm:metricmapgf}, 
\begin{align*}
	V_{g,n}^{(2\pi)}(2\pi,2\pi,\ldots) \neq 2^{2-2g-n} V_{g,n}^{\mathrm{WP}}(0,0,\ldots)\qquad\text{for }g\geq 2
\end{align*}
but their generating functions are closely related.
It is natural to ask whether one can understand their relation from the perspective of the fuchsian ideal polyhedra.

\subsection*{Acknowledgments}
This work is part of the START-UP 2018 programme with project number 740.018.017, which is financed by the Dutch Research Council (NWO).

\section{From discrete to metric maps}\label{sec:discrete}

We start by recalling some definitions about irreducibility of (discrete) maps from \cite{Bouttier2014,Budd2020} .
A planar map $\map$ is said to be \emph{$d$-irreducible} for an integer $d\geq 0$ if all simple cycles have length at least $d$, i.e.\ $\map$ has girth at least $d$, and the only simple cycles of length $d$ are the contours of the faces of degree $d$.
A genus-$g$ map $\map$ for $g\geq 1$ is \emph{essentially $d$-irreducible} if its universal cover, seen as an infinite planar map, is $d$-irreducible.
We will only be dealing with \emph{even maps}, meaning that all faces have even degree.
In this case one may restrict attention to the criterion of essential $2b$-irreducibility for an integer $b\geq 0$.
Note that every map is essentially $0$-irreducible.
See Figure \ref{fig:irreducibleexamples} for examples.

\begin{figure}
	\centering
	\includegraphics[width=.6\linewidth]{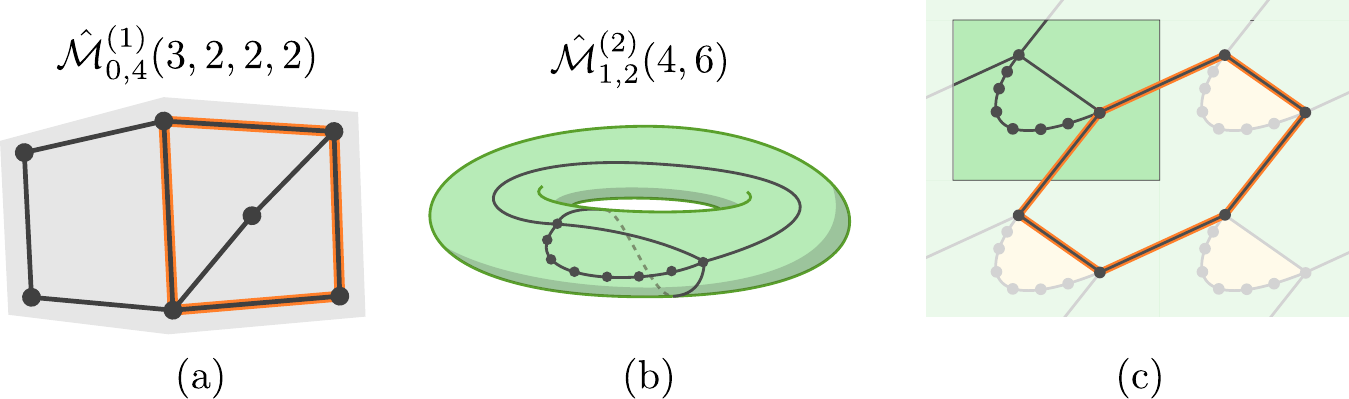}
	\caption{(a) A $2$-irreducible planar map that is not $4$-irreducible. (b) An essentially $4$-irreducible genus-$1$ map that is not $6$-irreducible, because it contains an essentially simple cycle of length $6$ that surrounds two faces, as illustrated in the universal cover (c). \label{fig:irreducibleexamples} }
\end{figure}

\subsection{Relation between polynomials}

For $g\geq 0$, $b\geq 1$ and $n\geq 1$ (provided $n\geq 3$ if $g=0$) and $\ell_1,\ldots,\ell_n \geq b$ we denote by $\hat{\mathcal{M}}_{g,n}^{(b)}(\ell_1,\ldots,\ell_n)$ the set of essentially $2b$-irreducible genus-$g$ maps with $n$ labeled faces of degrees $2\ell_1,\ldots,2\ell_n$ and no vertices of degree one.
Let $\hat{\vec{\mathcal{M}}}_{g,n}^{(b)}(\ell_1,\ldots,\ell_n)$ be the corresponding set of rooted maps.
We enumerate this set via
\begin{equation*}
	\|\hat{\mathcal{M}}_{g,n}^{(b)}(\ell_1,\ldots,\ell_n)\| \coloneqq \sum_{\map\in\hat{\mathcal{M}}_{g,n}^{(b)}(\ell_1,\ldots,\ell_n) } \frac{1}{|\operatorname{Aut}(\map)|} = \sum_{\map\in\hat{\vec{\mathcal{M}}}_{g,n}^{(b)}(\ell_1,\ldots,\ell_n) } \frac{1}{2|\mathcal{E}(\map)|}
\end{equation*}
where $\operatorname{Aut}(\map)$ is the group of orientation-preserving automorphisms of the map $\map$ that preserve the face labels.
According to \cite[Theorem 1]{Budd2020}, for each $g\geq0$ and $n\geq 1$ (and $n\geq 3$ if $g=0$) there exists a polynomial $\hat{N}^{(b)}_{g,n}(\ell_1,\ldots,\ell_n)$ of degree $2n+6g-6$ in $b,\ell_1,\ldots,\ell_n$, which is even and symmetric in $\ell_1,\ldots,\ell_n$, such that
\begin{equation*}
	\|\hat{\mathcal{M}}_{g,n}^{(b)}(\ell_1,\ldots,\ell_n)\| = \hat{N}^{(b)}_{g,n}(\ell_1,\ldots,\ell_n)
\end{equation*}
for all $\ell_1,\ldots,\ell_n > b \geq 1$.
For convenience we are restricting our attention to the situation where the face degrees are strictly larger than $b$, so that we do not have to deal with the small correction in the case $\ell_1=\cdots=\ell_n=b$ and $g=0$ (see \cite[Equation (1)]{Budd2020}).
With this restriction $\hat{N}^{(b)}_{g,n}(\ell_1,\ldots,\ell_n)$ equivalently counts genus-$g$ maps with essential girth at least $2b+2$. 
In the limit that we will consider below in Lemma \ref{lem:vague} the distinction between essential $2b$-irreducibility and having essential girth at least $2b+2$ does not survive, and as remarked at the end of Section~\ref{lem:vague} the continuum analogues of these notions give rise to the same volumes.

Let $\map$ be a genus-$g$ map with $n\geq 1$ labeled faces ($n\geq 3$ in case $g=0$) of degrees at least $2b$ and no vertices of degree one. 
A closed path on $\map$ is said to be an \emph{essentially simple cycle} if it lifts to a simple cycle in $\map^\infty$ that encloses at most one lift of each face in $\map$.
In \cite[Section 2]{Budd2020}  we showed that $\map$ is essentially $2b$-irreducible if and only if each essentially simple cycle of $\map$ that encloses at least two faces (on both sides in the case $g=0$) has length strictly larger than $2b$.
This is a practical criterion since a map contains only finitely many essentially simple cycles (as opposed to the infinitely many simple cycles in its universal cover).

To understand how discrete maps and metric maps are related, we employ a method akin to that of \cite[Section 3]{Ambjorn_Multi_2016}.
The \emph{skeleton} of a genus-$g$ map $\map$ without vertices of degree one is the map $\mathsf{Skel}(\map)$ obtained by deleting each vertex of degree two and merging its incident edges.
If $\map$ is rooted, we take the result to be rooted on the edge into which the root edge of $\map$ was merged (with the same orientation).
For $a \in (0,\infty)$ we let $\mathsf{Skel}_{a}(\map)$ be the metric map obtained from $\mathsf{Skel}(\map)$ by assigning length $\tfrac{a}{2}k$ to an edge of the skeleton that corresponds to a chain of $k$ edges in $\map$.
Then for $b \geq 1$ and $\beta \in(0,\infty)$ the metric map $\mathsf{Skel}_{\beta/b}(\map)$ is essentially $\beta$-irreducible if and only if $\map$ is essentially $2b$-irreducible.
Moreover,
\begin{equation*}
	\hat{\vec{\mathcal{M}}}_{g,n}^{\star(b)} \coloneqq \{ \map\in \hat{\vec{\mathcal{M}}}_{g,n}^{(b)}: \text{ root degree at least }3\} \xrightarrow{\mathsf{Skel}_{\beta/b}} \vec{\mathcal{R}}_{g,n}^{(\beta)}
\end{equation*}
is injective, where the \emph{root degree} of $\map$ is the degree of the vertex at the origin of the root edge.
We may use it to introduce a discrete measure $\mu_b$ on $\vec{\mathcal{R}}_{g,n}^{(\beta)}$ that assigns to each metric map $\mathfrak{s}$ in the image a weight $1/(2|\mathcal{E}(\mathfrak{s})|)$,
\begin{equation*}
	\mu_b = \sum_{\map\in \hat{\vec{\mathcal{M}}}_{g,n}^{\star(b)}} \frac{\delta_{\mathsf{Skel}_{\beta/b}(\map)}}{2|\mathcal{E}(\mathsf{Skel}(\map))|}.
\end{equation*}

\begin{lemma}\label{lem:vague}
As $b\to\infty$ the measure $\mu_b$ converges after rescaling to the measure $\mu_{g,n}$ on $\vec{\mathcal{R}}_{g,n}^{(\beta)}$ defined in \eqref{eq:measuremu},
\begin{equation*}
	2^{n-1}\left(\frac{\beta}{2b}\right)^{6g-6+3n} \, \mu_b \xrightarrow[b\to\infty]{\text{\quad vague\quad}} \mu_{g,n}.
\end{equation*}
\end{lemma}
\begin{proof}
	Fix a rooted genus-$g$ map $\mathfrak{s}$ with $n$ labeled faces and $|\mathcal{E}(\mathfrak{s})|=k$ edges and all vertices of degree at least three.
	Comparing $\mu_b$ to the definition \eqref{eq:measuremu} of $\mu_{g,n}$ it is sufficient to prove that 
	\begin{equation}\label{eq:skelconv}
		2^{n-1}\left(\frac{\beta}{2b}\right)^{6g-6+3n} \, \sum_{\map} \delta_{\mathsf{Skel}_{\beta/b}(\map)} \xrightarrow[b\to\infty]{\text{\quad vague\quad}} \begin{cases}
			\mu_{\mathfrak{s}} &\text{if }\mathfrak{s}\text{ is cubic} \\ 0 & \text{otherwise}
		\end{cases},
	\end{equation}
	where the sum ranges over the maps $\map\in \hat{\vec{\mathcal{M}}}_{g,n}^{\star(b)}$ that have $\mathsf{Skel}(\map) = \mathfrak{s}$.
	
	Denote the edges of $\mathfrak{s}$ by $e_1,\ldots,e_k$.
	Let $c^1,\ldots,c^p$ be a complete list of essentially simple cycles of $\mathfrak{s}$ and for each $i=1,\ldots,p$ let $\mathbf{c}^i \in \{0,1,2\}^k$ be the vector counting the number of occurrences of each edge $e_1,\ldots,e_k$ in $c^i$.
	Let $P_{\mathfrak{s}} \subset \R^k$ be the convex open polytope defined by 
	\begin{equation*}
		P_{\mathfrak{s}} = \left\{ \mathbf{x}\in\R_{>0}^k : \mathbf{c}^i\cdot \mathbf{x} > \beta \text{ for }i=1,\ldots,p\right\}.
	\end{equation*}
	This is the polytope of length assignments mentioned in the introduction.
	If $\mathfrak{s}$ is cubic then $\mu_{\mathfrak{s}}$ is the push-forward of the Lebesgue measure on $P_{\mathfrak{s}}$ to $\vec{\mathcal{R}}_{g,n}^{(\beta)}$.
	
	It should be clear that the set of all rooted genus-$g$ maps $\map$ with root degree at least three and skeleton $\mathsf{Skel}(\map) = \mathfrak{s}$ is in bijection with positive integer vectors $\mathbf{x}=(x_1,\ldots,x_k)\in \Z_{>0}^k$ by taking $x_i$ to be the number of edges of $\map$ that are merged into the edge $e_i$ (see \cite[Section 2]{Budd2020} for a more thorough discussion).
	Denote by $\mathbf{f}^j\in \{0,1,2\}^k$ for each $j=1,\ldots,n$ the number of occurrences of the edges $e_1,\ldots,e_k$ in the contour of the face labeled $j$.
	Then the set of metric maps
	\begin{equation*}\mathsf{Skel}_{\beta/b}\Big(\Big\{\map\in \hat{\vec{\mathcal{M}}}_{g,n}^{\star(b)} : \mathsf{Skel}(\map) = \mathfrak{s} \Big\}\Big)
	\end{equation*}
	corresponds to the set of length assignments
	\begin{equation*}
		P_{\mathfrak{s}} \cap 
		\left\{ \tfrac{\beta}{2b}\mathbf{x}: \mathbf{x}\in\Z_{>0}^k\text{ and }\mathbf{f}^j\cdot \mathbf{x}\in 2\Z\text{ for }j=1,\ldots,n \right\}.
	\end{equation*}
	Since we always have $\sum_{j=1}^n \mathbf{f}^j\in (2\Z)^k$, this amounts to $2^{n-1}$ independent parity conditions on $\mathbf{x}$.
	In particular, any hypercube $\mathbf{y} + [0,\tfrac{\beta}{b})^{k}$, $\mathbf{y} \in \R^k$, that is fully contained in $P_{\mathfrak{s}}$ contains exactly $2^{k+1-n}$ elements.
	Hence, as $b\to\infty$ the counting measure multiplied by $2^{n-1}\left(\frac{\beta}{2b}\right)^{6g-6+3n}$ converges to the Lebesgue measure if $k = 6g-6+3n$ and to $0$ if $k < 6g-g+3n$.
	These two cases correspond precisely to $\mathfrak{s}$ being cubic or not, thus verifying \eqref{eq:skelconv}.
\end{proof}

\noindent
Recall that $\mathsf{Circ} : \vec{\mathcal{R}}^{(\beta)}_{n,g} \to [\beta, \infty)^n$ is the mapping that extracts the circumferences of the $n$ labeled faces.

\begin{lemma}\label{lem:discretemeasure}
	The push-forward measure $\mathsf{Circ}_*\mu_b$ on $[\beta, \infty)^n$ is given by the discrete measure
	\begin{align*}
		\mathsf{Circ}_*\mu_b = \sum_{\ell_1, \ldots, \ell_n > b} \hat{N}_{g,n}^{(b)}(\ell_1,\ldots,\ell_n)\,\,\delta_{\left(\frac{\beta}{b} \ell_1, \ldots, \frac{\beta}{b} \ell_n\right)}
	\end{align*}
\end{lemma}
\begin{proof}
	For any $\ell_1,\ldots,\ell_n > b \geq 1$ we compute from the definition of $\mu_b$ that
	\begin{align*}
		\mathsf{Circ}_*\mu_b\Big(\Big\{\Big(\tfrac{\beta}{b} \ell_1, \ldots, \tfrac{\beta}{b} \ell_n\Big)\Big\}\Big) &= \sum_{\map\in\hat{\vec{\mathcal{M}}}_{g,n}^{(b)}(\ell_1,\ldots,\ell_n) } \frac{\ind_{\{\text{root degree of }\map\text{ is at least }3\}}}{2|\mathcal{E}(\mathsf{Skel}(\map))|}= \sum_{\map\in\hat{\vec{\mathcal{M}}}_{g,n}^{(b)}(\ell_1,\ldots,\ell_n) } \frac{1}{2|\mathcal{E}(\map)|}\\
		&= \|\hat{\mathcal{M}}_{g,n}^{(b)}(\ell_1,\ldots,\ell_n)\| = \hat{N}^{(b)}_{g,n}(\ell_1,\ldots,\ell_n).
	\end{align*}
\end{proof}

\noindent
We now have all ingredients to relate the volumes $V_{g,n}^{(\beta)}$ to the polynomials $\hat{N}_{g,n}^{(b)}$ counting essentially $2b$-irreducible maps.

\begin{proposition}\label{thm:polynomiallimit}
	For $g\geq 0$ and $n\geq 1$ (provided $n\geq 3$ if $g=0$), $V_{g,n}^{(\beta)}(\alpha_1,\ldots,\alpha_n)$ is a symmetric polynomial in $\alpha_1^2,\ldots,\alpha_n^2$ of $3g-3+n$ and a homogeneous polynomial of degree $3g-3+n$ in $\alpha_1^2,\ldots,\alpha_n^2,\beta^2$ and is explicitly given by
	\begin{equation}\label{eq:densitypol}
		V_{g,n}^{(\beta)}(\alpha_1,\ldots,\alpha_n) = \frac{1}{2}[t^{6g-6+2n}]
		\hat{N}_{g,n}^{(\tfrac12\beta t)}(\tfrac12t\alpha_1,\ldots,\tfrac12t\alpha_n).
		\end{equation}
\end{proposition}
\begin{proof} 
Since $\hat{N}_{g,n}^{(b)}(\ell_1,\ldots,\ell_n)$ is a polynomial of degree $6g-6+2n$ in $b,\ell_1,\ldots,\ell_n$, we have
\begin{align*}
	2^{n-1}\left(\frac{\beta}{2b}\right)^{6g-6+3n} \left(\frac{b}{\beta}\right)^{n}
	\hat{N}_{g,n}^{(b)}(\tfrac{b}{\beta}\alpha_1,\ldots,\tfrac{b}{\beta}\alpha_n) \xrightarrow{b\to\infty} \tfrac12 [t^{6g-6+2n}]\hat{N}_{g,n}^{(\tfrac12\beta t)}(\tfrac12\alpha_1t,\ldots,\tfrac12\alpha_nt).
	\end{align*}
	Together with Lemma \ref{lem:discretemeasure} this implies the vague convergence of measures on $[\beta,\infty)^n$,
	\begin{equation*}
		2^{n-1}\left(\frac{\beta}{2b}\right)^{6g-6+3n}\mathsf{Circ}_*\mu_b \xrightarrow[b\to\infty]{\text{\quad vague\quad}}  \tfrac12 [t^{6g-6+2n}]\hat{N}_{g,n}^{(\tfrac12\beta t)}(\tfrac12\alpha_1t,\ldots,\tfrac12\alpha_nt)\,\rmd\alpha_1\cdots\rmd\alpha_n. 
	\end{equation*}
	On the other hand the mapping $\mathsf{Circ} : \vec{\mathcal{R}}^{(\beta)}_{n,g} \to [\beta, \infty)^n$ is continuous, so it follows from Lemma \ref{lem:vague} that we also have the vague convergence
	\begin{equation*}
		2^{n-1}\left(\frac{\beta}{2b}\right)^{6g-6+3n}\mathsf{Circ}_*\mu_b \xrightarrow[b\to\infty]{\text{\quad vague\quad}} \mathsf{Circ}_*\mu_{g,n}.
	\end{equation*}
	So $\mathsf{Circ}_*\mu_{g,n}$ indeed has a density with respect to the Lebesgue measure $\rmd\alpha_1\cdots\rmd\alpha_n$ given by \eqref{eq:densitypol}.
	Since $\hat{N}_{g,n}^{(b)}(\ell_1,\ldots,\ell_n)$ is even in $\ell_1,\ldots,\ell_n$ and we are extracting a homogeneous part of even degree, $V_{g,n}^{(\beta)}(\alpha_1,\ldots,\alpha_n)$ is necessarily even in $\beta,\alpha_1,\ldots,\alpha_n$.
\end{proof}

\subsection{Proof of Theorem~\ref{thm:realmaps}}

The first statement of Theorem~\ref{thm:realmaps} is proved in the Proposition above.
The fact that $V_{g,1}^{(\beta)}(\alpha_1)$ is a rational multiple of $\alpha_1^{6g-4}$ is a direct consequence of the analogous statement in \cite[Theorem 2]{Budd2020} that $\hat{N}_{g,1}^{(b)}(\ell_1)$ is independent of $b$.

It remains to verify the string and dilaton equations.
According to \cite[Theorem 2]{Budd2020} the polynomials $\hat{N}_{g,n}^{(b)}$ satisfy the string equation
\begin{equation}\label{eq:discretestring} 
\hat{N}^{(b)}_{g,n+1}(\ell_1,\ldots,\ell_n,1) = \sum_{j=1}^n \left(\sum_{k=b+1}^{\ell_j} 2 k\, \hat{N}^{(b)}_{g,n}(\ell_1,\ldots,\ell_{j-1},k,\ell_{j+1},\ldots,\ell_n) - \ell_j \hat{N}^{(b)}_{g,n}(\ell_1,\ldots,\ell_n)\right)
\end{equation}
and the dilaton equation
\begin{equation}\label{eq:discretedilaton} 
\hat{N}_{g,n+1}^{(b)}(\ell_1,\ldots,\ell_n,1) - \hat{N}_{g,n+1}^{(b)}(\ell_1,\ldots,\ell_n,0) = (n+2g-2)\, \hat{N}_{g,n}^{(b)}(\ell_1,\ldots,\ell_n).
\end{equation}
For any integer $p\geq 0$,
\begin{equation*}
\left(-\ell^{2p+1}+\sum_{k=b+1}^\ell 2k\,k^{2p}\right)  - \int_{b}^\ell 2k\,k^{2p} \rmd k
\end{equation*}
is a polynomial in $\ell$ of degree less than $2p+2$.
Hence, the string equation implies that
\begin{align*}
V_{g,n+1}^{(\beta)}(\alpha_1,\ldots,\alpha_n,0) &= \tfrac12 [t^{6g-4+2n}] 
\hat{N}_{g,n+1}^{(\tfrac12\beta t)}(\tfrac12t\alpha_1,\ldots,\tfrac12t\alpha_n,1).\\
&= \tfrac12\sum_{j=1}^n [t^{6g-4+2n}]\int_{\tfrac12t\beta}^{\tfrac12t\alpha_j} 2k\, \hat{N}^{(\tfrac12t\beta)}_{g,n}(\tfrac12t\alpha_1,\ldots,\tfrac12t\alpha_{j-1},k,\tfrac12t\alpha_{j+1},\ldots,\tfrac12t\alpha_n)\rmd k\\
&=\tfrac14 \sum_{j=1}^n\int_{\beta}^{\alpha_j} \alpha\, [t^{6g-6+2n}] \hat{N}^{(\tfrac12t\beta)}_{g,n}(\tfrac12t\alpha_1,\ldots,\tfrac12t\alpha_{j-1},\tfrac12t\alpha,\tfrac12t\alpha_{j+1},\ldots,\tfrac12t\alpha_n)\rmd \alpha\\
&=\tfrac12 \sum_{j=1}^n\int_{\beta}^{\alpha_j} \alpha_j V_{g,n}^{(\beta)}(\alpha_1,\ldots,\alpha_n).
\end{align*}
Similarly, by Taylor expansion
\begin{equation*}
\hat{N}_{g,n+1}^{(b)}(\ell_1,\ldots,\ell_n,1) - \hat{N}_{g,n+1}^{(b)}(\ell_1,\ldots,\ell_n,0) - \frac{1}{2} \frac{\partial^2\hat{N}_{g,n+1}^{(b)}}{\partial \ell_{n+1}^2} (\ell_1,\ldots,\ell_n,0)
\end{equation*}
is a polynomial of degree less than $6g-6+2n$ in $b,\ell_1,\ldots,\ell_n$.
Hence, the dilaton equation implies
\begin{align*}
\frac{\partial^2 V_{g,n+1}^{(\beta)}}{\partial\alpha_{n+1}^2}(\alpha_1,\ldots,\alpha_n,0) &= \frac{1}{8} [t^{6g-4+2n}] t^2 \frac{\partial^2\hat{N}_{g,n+1}^{(\tfrac12t\beta)}}{\partial \ell_{n+1}^2} (\tfrac12t\alpha_1,\ldots,\tfrac12t\alpha_n,0) \\
&= \frac{1}{4}(2g-2+n) [t^{6g-6+2n}]\hat{N}^{(\tfrac12t\beta)}_{g,n}(\tfrac12t\alpha_1,\ldots,\tfrac12t\alpha_n) \\ 
&=\tfrac12(2g-2+n) V_{g,n}^{(\beta)}(\alpha_1,\ldots,\alpha_n).
\end{align*}
This finishes the proof.

\section{Generating functions for irreducible metric maps}\label{sec:genfunirrmaps}

Let us fix an integer $d\geq 0$.
The goal of this section will be to derive expressions for the formal power series 
\begin{equation}\label{eq:partitionfunctionbeta}
F^{(\beta)}_g(\alpha_1,\ldots,\alpha_d; x_0,\ldots,x_d) = \sum_{n=1}^\infty\frac{1}{n!} \sum_{i_1=0}^d x_{i_1} \cdots\! \sum_{i_n=0}^d x_{i_n} V^{(\beta)}_{g,n}(\alpha_{i_1}\ldots,\alpha_{i_n}),
\end{equation}
where we use the convention $\alpha_0 = \beta$.
We already know that its coefficients are related to the polynomials $\hat{N}_{g,n}^{(b)}(\ell_1,\ldots,\ell_n)$ through the extraction of the homogeneous part of top degree, i.e. 
\begin{equation}\label{eq:polynomialrelation}
	V_{g,n}^{(\beta)}(\alpha_1,\ldots,\alpha_n) = \frac{1}{2}[t^{6g-6+2n}] 
	\hat{N}_{g,n}^{(\tfrac12\beta t)}(\tfrac12t\alpha_1,\ldots,\tfrac12t\alpha_n).
\end{equation}
Explicit expressions for the formal generating functions of these polynomials were obtained in \cite{Budd2020}, so it remains to carefully remove all sub-leading contributions.
The main result of this section, whose proof appears in Section~\ref{sec:proofmetricmapgf} below, can be formulated as follows.

\begin{proposition}\label{thm:metricmapgf}
Let
\begin{equation}\label{eq:Zbeta}
	Z^{(\beta)}(r) = \frac{2\sqrt{r}}{\beta} J_1(\beta\sqrt{r}) - \sum_{i=0}^d I_0\left( \sqrt{(\alpha_i^2-\beta^2)r}\right)\,x_i,
\end{equation}
where $I_0$ and $J_1$ are (modified) Bessel functions of the first kind, and let $R^{(\beta)}$ be the formal power series solution to $Z^{(\beta)}(R^{(\beta)}) = 0$.
Define the moments $M^{(\beta)}_p$ recursively via
\begin{equation}\label{eq:Mbeta}
	M^{(\beta)}_0 = \frac{1}{\partial_{x_0} R^{(\beta)}}, \qquad M^{(\beta)}_p =  \left(\tfrac{1}{p!}\left(\tfrac{\beta}{2}\right)^{2p}+\partial_{x_0} M^{(\beta)}_{p-1}\right)M^{(\beta)}_0,\qquad p\geq 1.
\end{equation}
Then the partition functions $F_g^{(\beta)}$ are given by
\begin{align}
	F^{(\beta)}_0 &= \frac{1}{4}\int_{0}^{R^{(\beta)}} \!\!\!\rmd r\, \left(Z^{(\beta)}(r)\right)^2,\label{eq:genfun0beta}\\
	F^{(\beta)}_1 &= - \frac{1}{24} \log M^{(\beta)}_0,\label{eq:genfun1beta}\\
	F^{(\beta)}_g &= \left(\frac{2}{(M_0^{(\beta)})^2}\right)^{g-1}\bar{\mathcal{P}}_g\left(\frac{M^{(\beta)}_1}{M^{(\beta)}_0},\ldots, \frac{M^{(\beta)}_{3g-3}}{M^{(\beta)}_0}\right) \quad \text{for }g\geq 2,\label{eq:genfun2beta}
	\end{align}
where for each $g\geq 2$ the polynomial $\bar{\mathcal{P}}_g(m_1,\ldots,m_{3g-3})$ is such that $\bar{\mathcal{P}}_g(m_1 t,m_2 t^2,\ldots,m_{3g-3}t^{3g-3})$ is homogeneous of degree $3g-3$ in $t$.
\end{proposition}

Note that upon setting $\beta=2\pi$ this yields the statement of Theorem~\ref{thm:genfun} in the case of essentially $2\pi$-irreducible metric maps, except that the polynomials $\bar{\mathcal{P}}_g$ have not yet been identified with $\mathcal{P}_g$ in \eqref{eq:Ppol}.
This will be taken care of in Section~\ref{sec:proofgenfun} below. 

\subsection{Summary of generating functions for discrete maps}

Our derivation of the generating functions $F_g^{(\beta)}$ is based on the discrete analogues obtained in \cite{Budd2020}.
We start by reformulating the results of \cite{Budd2020} in a form that is convenient for our current purposes.
Consider the formal generating function $\hat{F}^{(b)}_g$ of the polynomials $\hat{N}_{g,n}^{(b)}$ defined via
\begin{align}
	\hat{F}^{(b)}_g(\ell_1,\ldots,\ell_d;\hat{x}_0,\ldots,\hat{x}_d) &= \sum_{n=1}^\infty\frac{1}{n!} \sum_{i_1=0}^d \hat{x}_{i_1} \cdots\! \sum_{i_n=0}^d \hat{x}_{i_n} \hat{N}_{g,n}^{(b)}(\ell_{i_1},\ldots,\ell_{i_n}),\label{eq:discretepartitionfunction}
\end{align}
where we use the convention $\ell_0 = b$. 
We may view it as an element of $\Q[b,\ell_1,\ldots,\ell_d][\![ \hat{x}_0, \ldots, \hat{x}_d]\!]$, i.e.\ as a formal power series in $\hat{x}_0,\ldots,\hat{x}_d$ with coefficients that are polynomials in $b,\ell_1,\ldots,\ell_d$.

Let $I(b,\ell;r) \in \Q[b,\ell][\![r]\!]$ and $J(b;r)\in \Q[b][\![r]\!]$ be the formal power series in $r$ defined by
\begin{align}
	I(b,\ell;r) &= \sum_{p=0}^\infty \frac{r^p}{(p!)^2}\prod_{m=0}^{p-1}(\ell^2-(b-m)^2) = {_2F_1}(\ell-b,-\ell-b;1;-r), \label{eq:Iseries}\\
	J(b;r) &= \sum_{p=1}^\infty \frac{(-1)^{p+1} r^p}{p!(p-1)!} \prod_{m=0}^{p-2} (b-m)(b-m-1) = r\,\,{_2F_1}(1-b,-b;2;-r),\label{eq:Jseries}
\end{align}
where ${_2F_1}$ is the hypergeometric function ${_2F_1}(a,b;c;z) = \sum_{n=0}^\infty \frac{(a)_n(b)_n}{(c)_n}\frac{z^n}{n!}$.

The expressions for $\hat{F}^{(b)}_g$ involve a formal power series $\hat{R}^{(b)}(\ell_1,\ldots,\ell_d;\hat{x}_0,\ldots,\hat{x}_d)$ that is defined to be the formal solution to
\begin{equation*}
\hat{Z}^{(b)}(\hat{R}^{(b)}) = 0, \qquad \hat{Z}^{(b)}(r)\equiv \hat{Z}^{(b)}(\ell_1,\ldots,\ell_d;\hat{x}_0,\ldots,\hat{x}_d; r) \coloneqq J(b,r) - \sum_{i=0}^d I(b,\ell_i;r)\,\hat{x}_i,
\end{equation*}
such that $R^{(b)}(0,\ldots,0;\ell_1,\ldots,\ell_d) = 0$.

It is shown in \cite[Lemma 7]{Budd2020} that for each $p\geq 0$ there exists a unique polynomial $Q_p(b,j)$ of degree $2p+1$ in $b$ and degree $p+1$ in $j$ such that
\begin{align}\label{eq:qdef}
\sum_{k=b}^{j-1} \binom{2k+1+p}{2p+1}\binom{2j-1}{j+k} &= \binom{2j-1}{j+b} Q_p(b,j) \qquad \text{for }j > b \geq 0.
\end{align}
With the help of these we can introduce the moments
\begin{equation}\label{eq:mhat}
	\hat{M}_p^{(b)}(\ell_1,\ldots,\ell_d;\hat{x}_1,\ldots,\hat{x}_d) = Q_p(b,(1+r)\, \partial_r)\,(1+r)^{-b} \, \hat{Z}^{(b)}(r) \Big|_{r=\hat{R}^{(b)}}.
\end{equation}
According to \cite[Lemma 13 \& Proposition 16]{Budd2020} both $\hat{R}^{(b)}$ and $\hat{M}^{(b)}_p$ for $p\geq 0$ may be interpreted as formal power series in $\hat{x}_1,\ldots,\hat{x}_d$ with coefficients that are polynomials in $b,\ell_1,\ldots,\ell_d$, just like $\hat{F}_g^{(b)}$.
Then \cite[Proposition 14-16]{Budd2020} states that 
\begin{align}
	\frac{\partial}{\partial \hat{x}_{1}}\frac{\partial}{\partial \hat{x}_{2}}\hat{F}_0^{(b)} &= \int_{0}^{\hat{R}^{(b)}} \!\!\!\rmd r\, \frac{I(b,\ell_1;r)I(b,\ell_2;r)}{(1+r)^{2b+1}}, \label{eq:Fhat0part}\\
	\hat{F}^{(b)}_1 &= - \frac{1}{12} \log \hat{M}^{(b)}_0, \label{eq:Fhat1part}\\
	\hat{F}^{(b)}_g &= P_g\left( \frac{1}{\hat{M}^{(b)}_0}, \frac{\hat{M}^{(b)}_1}{\hat{M}^{(b)}_0},\ldots, \frac{\hat{M}^{(b)}_{3g-3}}{\hat{M}^{(b)}_0}\right) \quad \text{for }g\geq 2.\label{eq:Fhatgpart}
	\end{align}
Note that the first expression uniquely determines the formal power series $\hat{F}_0^{(b)}$ since it is invariant under permutation of the pairs $(x_1,\ell_1), \ldots, (x_d,\ell_d)$ and its coefficients up to quadratic order in $x_i$ vanish.
Here  $P_g(m_0,\ldots,m_{3g-3})$ is a universal polynomial with rational coefficients for each $g\geq 2$.
More precisely, $P_g$ is of the form 
\begin{equation}\label{eq:Ptilde}
P_g(m_0,\ldots,m_{3g-3}) = \tilde{P}_g(0,\ldots,0) - m_0^{2g-2} \tilde{P}_g(m_1,\ldots,m_{3g-3})
\end{equation}
for some polynomial $\tilde{P}_g(m_1,\ldots,m_{3g-3})$ such that $\tilde{P}_g(\mu,\mu^2,\ldots,\mu^{3g-3})$ is of degree $3g-3$ in $\mu$.

\subsection{Extract leading orders in the coefficients}

Given a polynomial $f(X_1,\ldots,X_n)$, its homogeneous part of degree $k$ is
\begin{equation*}
	[t^k] f(tX_1,\ldots,tX_n).
\end{equation*}
It will be convenient for the exposition to introduce a linear operator that extracts the homogeneous part in the coefficients of a power series as follows.
For any $p\in\Z$ we let $\Omega_p$ be the linear operator on the ring of power series $\Q[X_1,\ldots,X_n][\![ Z_1,\ldots,Z_m ]\!]$ with polynomial coefficients extracting the homogeneous part of degree $p + 2(i_1 +\cdots i_m)$ from the coefficient of $Z_1^{i_1}\cdots Z_m^{i_m}$, i.e.
\begin{equation*}
	\Omega_p F = \sum_{i_1,\ldots,i_m} \big([t^{p+2(i_1+\cdots+i_m)}]\,[Z_1^{i_1}\cdots Z_m^{i_m}]F(t X_1, \ldots, t X_n; Z_1, \ldots, Z_m) \big)\,Z_1^{i_1}\cdots Z_m^{i_m}.
\end{equation*}
We say that such a power series is of \emph{order} $p$ if $\Omega_p F\neq 0$ and $\Omega_k F=0$ for all $k>p$.

From their definitions we directly verify that $I(b,\ell;r) \in \Q[b,\ell][\![r]\!]$ is of order $0$, $J(b;r) \in \Q[b][\![r]\!]$ is of order $-2$, and $\hat{Z}^{(b)} \in \Q[b,\ell_1,\ldots,\ell_d][\![\hat{x}_0,\ldots,\hat{x}_d,r]\!]$ is of order $-2$ as well.
Moreover, at leading order we have
\begin{align*}
	\Omega_0 I(b,\ell;r) &= \sum_{p=0}^\infty \frac{r^p}{(p!)^2} (\ell^2-b^2)^p = I_0\left(2\sqrt{(\ell^2-b^2)r}\right),\\ 
	\Omega_{-2} J(b,r) &= \sum_{p=1}^\infty \frac{(-1)^{p+1} r^p}{p!(p-1)!} b^{2p-2} = \frac{\sqrt{r}}{b}J_1(2b\sqrt{r}),\\
	\Omega_{-2}\hat{Z}^{(b)}(r)&= \frac{\sqrt{r}}{b}J_1(2b\sqrt{r}) - \sum_{i=0}^d I_0\left(2\sqrt{(\ell_i^2-b^2)r}\right)\, \hat{x}_i,
\end{align*}
where $I_0$ and $J_1$ are (modified) Bessel functions of the first kind.
It follows from \cite[Lemma 13]{Budd2020} that $\hat{R}^{(b)}\in \Q[b,\ell_1,\ldots,\ell_d][\![\hat{x}_0,\ldots,\hat{x}_d]\!]$ is of degree $-2$, so
\begin{equation*}
	\Omega_{-2} \hat{Z}^{(b)}( \hat{R}^{(b)}) = (\Omega_{-2} \hat{Z}^{(b)})( \Omega_{-2}\hat{R}^{(b)}) = 0.
\end{equation*}
Therefore $\Omega_{-2} \hat{R}^{(b)}$ is the unique formal power series solution to $\Omega_{-2} \hat{Z}^{(b)}(r) = 0$.

In order to find an expression for $M_p^{(\beta)}$ we need the following lemma.
\begin{lemma}
	For any $p \geq 0$ the polynomial $t\mapsto Q_p(t b,t^2 j)$ is of degree $2p+2$ and the top-degree coefficient is given by
	\begin{equation*}
	[t^{2p+2}] Q_p(t b,t^2 j) = 4^p\frac{p!}{(2p+1)!}\sum_{k=0}^p \frac{b^{2p-2k} j^{k+1}}{(p-k)!}.
	\end{equation*}
\end{lemma}
\begin{proof}
	Let us inspect the definition \eqref{eq:qdef} of $Q_p(b,j)$ that is valid for integers $j > b\geq 0$. 
	For positive real numbers $x,y$ we consider the summand in the limit $b,j,k\to \infty$ while maintaining the asymptotic ratios $j / b^2 \to x$ and $k / b \to y$.
	Stirling's formula easily gives
	\begin{equation*}
	\frac{\binom{2k+1+p}{2p+1}\binom{2j-1}{j+k}}{\binom{2j-1}{j+b}} = \frac{2^{2p+1}}{(2p+1)!} y^{2p+1} e^{\frac{1-y^2}{x}} b^{2p+1} (1 + O(b^{-1})).
	\end{equation*}
	Since this is integrable as $y\to\infty$, approximating the sum in \eqref{eq:qdef} by an integral we find that $Q_p(b,b^2 x)$ grows as $b^{2p+2}$ and 
	\begin{equation*}
	[b^{2p+2}]Q_p(b,b^2 x) = \frac{2^{2p+1}}{(2p+1)!} \int_1^\infty \rmd y\,  e^{\frac{1-y^2}{x}}y^{2p+1} =  4^p\frac{p!}{(2p+1)!}\sum_{k=0}^p \frac{x^{k+1}}{(p-k)!}.
	\end{equation*}
	Substituting $b \to tb$ and $x \to j / b^2$ gives the desired formula.
\end{proof}

With the help of this lemma and $\Omega_{-2}\,(1+r)^{-b} \, \hat{Z}^{(b)}(r) = \Omega_{-2}\,\hat{Z}^{(b)}(r)$ we find
\begin{align*}
	\Omega_{2p} Q_p(b,(1+r)\partial_r) \,(1+r)^{-b} \, \hat{Z}^{(b)}(r) &= ([t^{2p+2}]Q_p(t b,t^2 \partial_r))\Omega_{-2}\,(1+r)^{-b} \, \hat{Z}^{(b)}(r)\\
	&= 4^p\frac{p!}{(2p+1)!}\sum_{k=0}^p \frac{b^{2p-2k}}{(p-k)!} \partial_r^{k+1} \Omega_{-2}\hat{Z}^{(b)}(r).
\end{align*}
Hence $\hat{M}_p^{(b)}$ is of order $2p$ and satisfies
\begin{equation}\label{eq:mtop}
	\Omega_{2p}\hat{M}_p^{(b)} = 4^p\frac{p!}{(2p+1)!}\sum_{k=0}^p \frac{b^{2p-2k}}{(p-k)!} \left(\partial_r^{k+1} \Omega_{-2}\hat{Z}^{(b)}\right)(\Omega_{-2}\hat{R}^{(b)}).
\end{equation}
We are now in a position to extract the leading orders in the partition functions $\hat{F}_g^{(b)}$. 
Note that $\int \rmd r\, I(b,\ell_1;r)I(b,\ell_2;r) / (1+r)^{2b+1} \in \Q[b,\ell_1,\ell_2][\![r]\!]$ is of order $-2$ and 
\begin{align*}
	\Omega_{-2} \,\int \!\!\!\rmd r\, \frac{I(b,\ell_1;r)I(b,\ell_2;r)}{(1+r)^{2b+1}} &= \int \!\!\!\rmd r\, \Omega_0I(b,\ell_1;r)\,\Omega_0I(b,\ell_2;r).
\end{align*}
Therefore $\hat{F}_0^{(b)}$ is of order $-6$ and 
\begin{align}
	\frac{\partial}{\partial \hat{x}_{1}}\frac{\partial}{\partial \hat{x}_{2}} \Omega_{-6}\hat{F}_0^{(b)}= \Omega_{-2}\frac{\partial}{\partial \hat{x}_{1}}\frac{\partial}{\partial \hat{x}_{2}} \hat{F}_0^{(b)} &= \int_{0}^{\Omega_{-2}\hat{R}^{(b)}} \!\!\!\rmd r\, \Omega_0I(b,\ell_1;r)\,\Omega_0I(b,\ell_2;r).\label{eq:genfun0top}
\end{align}
Clearly $\hat{F}_1^{(b)}$ is of order $0$ and satisfies
\begin{equation}
	\Omega_0 \,\hat{F}_1^{(b)} = -\frac{1}{12}\log \Omega_0 \,\hat{M}_0^{(b)}.\label{eq:genfun1top}
\end{equation}
For $g\geq 2$ we find with the help of \eqref{eq:Ptilde} that $\hat{F}_g^{(b)}$ is of order $6g-6$ and
\begin{align}
	\Omega_{6g-6} \hat{F}_g^{(b)} &= - \Omega_{6g-6} \, (\hat{M}_0^{(b)})^{2-2g} \tilde{P}_g \left(\frac{\hat{M}_1^{(b)}}{\hat{M}_0^{(b)}}, \ldots,\frac{\hat{M}_{3g-3}^{(b)}}{\hat{M}_0^{(b)}}\right)\nonumber\\
	&= (\Omega_0\hat{M}_0^{(b)})^{2-2g} \bar{P}_g \left(\frac{\Omega_2\hat{M}_1^{(b)}}{\Omega_0\hat{M}_0^{(b)}}, \ldots,\frac{\Omega_{6g-6}\hat{M}_{3g-3}^{(b)}}{\Omega_0\hat{M}_0^{(b)}}\right),\label{eq:genfun2top}
\end{align}
where 
\begin{equation*}
\bar{P}_g(m_1,\ldots,m_{3g-3}) = -[t^{3g-3}] \tilde{P}_g(m_1 t,m_2t^2,\ldots,m_{3g-3}t^{3g-3}).	
\end{equation*}

\subsection{Proof of Proposition \ref{thm:metricmapgf}}\label{sec:proofmetricmapgf}
The combination of \eqref{eq:partitionfunctionbeta}, \eqref{eq:polynomialrelation} and \eqref{eq:discretepartitionfunction} translates into
\begin{equation}\label{eq:FgmetricFgtop}
	F_g^{(\beta)}(\alpha_1,\ldots,\alpha_d;x_0,\ldots,x_d) = \frac{1}{2} \Omega_{6g-6} \hat{F}_g^{(b)}\Big|_{b=\frac12\beta,\,\ell_i=\frac12\alpha_i,\,\hat{x}_i,=x_i}, \qquad g\geq 0.
\end{equation}
It therefore makes sense to introduce the following power series,
\begin{align*}
	I^{(\beta)}(\alpha; r) &= \Omega_0 I(b,\ell;r) \Big|_{b=\frac12\beta,\,\ell=\frac12\alpha} = I_0( \sqrt{(\alpha^2-\beta^2)r}),\\
	J^{(\beta)}(r) &= \Omega_{-2} J(b;r)\Big|_{b=\frac12\beta} = \frac{2\sqrt{r}}{\beta} J_1(\beta\sqrt{r}).
\end{align*}
Then $Z^{(\beta)}(r)$ and $R^{(\beta)}$ defined in \eqref{eq:Zbeta} are related to $\hat{Z}^{(b)}$ and $\hat{R}^{(b)}$ via 
\begin{align*}
	Z^{(\beta)}(r) &= \Omega_{-2} \hat{Z}^{(b)}\Big|_{b=\frac12\beta,\,\ell_i=\frac12\alpha_i,\, \hat{x}_i = x_i} = J^{(\beta)}(r) - \sum_{i=0}^d I^{(\beta)}(\alpha_i; r)\,x_i,\\
	R^{(\beta)} &= \Omega_{-2} \hat{R}^{(b)}\Big|_{b=\frac12\beta,\,\ell_i=\frac12\alpha_i,\, \hat{x}_i = x_i}.
\end{align*}
Next we verify that 
\begin{align*}
M^{(\beta)}_p &= 4^{-p} \frac{(2p+1)!}{p!} \Omega_{2p} \hat{M}_p^{(b)}\Big|_{b=\frac12\beta,\,\ell_i=\frac12\alpha_i,\, \hat{x}_i = x_i}
\end{align*}
satisfy the relations \eqref{eq:Mbeta}.
Rewriting \eqref{eq:mtop} in terms $Z^{(\beta)}$ and $R^{(\beta)}$ yields the expression
\begin{equation*}
	M^{(\beta)}_p = \sum_{k=0}^p \frac{(\beta/2)^{2p-2k}}{(p-k)!} \,\frac{\partial^{k+1} Z^{(\beta)}}{\partial r^{k+1}}( R^{(\beta)}).
\end{equation*}
Taking the $x_0$-derivative of $Z^{(\beta)}(R^{(\beta)})= 0$, and recalling that $\alpha_0=\beta$, we observe that $\frac{\partial Z^{(\beta)}}{\partial r}( R^{(\beta)}) \partial_{x_0} R^{(\beta)} = 1$.
Hence, 
\begin{equation*}
	M_0^{(\beta)} = \frac{\partial Z^{(\beta)}}{\partial r}( R^{(\beta)}) = \frac{1}{\partial_{x_0} R^{(\beta)}},
\end{equation*}
while for $p\geq 1$ we find 
\begin{align*}
\partial_{x_0} M^{(\beta)}_{p-1} &= \sum_{k=1}^{p} \frac{(\beta/2)^{2p-2k}}{(p-k)!} \,\partial_{x_0}\frac{\partial^{k} Z^{(\beta)}}{\partial r^{k}}( R^{(\beta)})\\
&= \sum_{k=1}^{p} \frac{(\beta/2)^{2p-2k}}{(p-k)!} \,\frac{\partial^{k+1} Z^{(\beta)}}{\partial r^{k+1}}( R^{(\beta)})\, \partial_{x_0} R^{(\beta)}.\\
&= M^{(\beta)}_p \partial_{x_0} R^{(\beta)} - \frac{1}{p!}\left(\frac{\beta}{2}\right)^{2p},
\end{align*}
which is in agreement with \eqref{eq:Mbeta}.

We have all ingredients to evaluate \eqref{eq:FgmetricFgtop}. 
For $g=0$, \eqref{eq:genfun0top} implies
\begin{equation*}
	\frac{\partial}{\partial x_{1}}\frac{\partial}{\partial x_{2}} F^{(\beta)}_0= \frac{1}{2} \int_{0}^{R^{(\beta)}} \!\!\!\rmd r\, I^{(\beta)}(\alpha_1;r)\,I^{(\beta)}(\alpha_2;r).
\end{equation*}
Using that $Z^{(\beta)}(R^{(\beta)}) = 0$ and that the coefficients of the linear and quadratic terms of $F_0^{(\beta)}$ vanish, it is easily seen that this expression integrates to 
\begin{equation*}
	F^{(\beta)}_0= \frac{1}{4} \int_{0}^{R^{(\beta)}} \!\!\!\rmd r\, \left(Z^{(\beta)}(r)\right)^2.
\end{equation*}
The expressions for $g=1$ follows directly from \eqref{eq:genfun1top}.
For $g\geq 2$ we introduce the polynomial $\bar{\mathcal{P}}_g$ via
\begin{equation*}
	\bar{\mathcal{P}}_g(m_1,\ldots,m_{3g-3}) = 2^{-g} \bar{P}_g(\tilde{m}_1,\ldots,\tilde{m}_{3g-3}), \qquad \tilde{m}_p = 4^p \frac{p!}{(2p+1)!} m_p.
\end{equation*}
Then the claimed expression for $g\geq 2$ follows from \eqref{eq:genfun2top}.
This finishes the proof.

\subsection{Expansion of the moments}
For reference we compute the initial terms of the power series $M_p^{(\beta)}$.
\begin{lemma}\label{lem:Mexpansion}
Up to linear order the moments $M_p^{(\beta)}$ are independent of $\beta$ and given by 
\begin{equation*}
	M_p^{(\beta)} = \ind_{p=0} - \frac{4^{-p-1}}{(p+1)!} \sum_{i=1}^d x_i\, \alpha_i^{2p+2} + \cdots.
\end{equation*}
\end{lemma}
\begin{proof}
Up to linear order
\begin{align*}
	\frac{\partial^{k+1} Z^{(\beta)}}{\partial r^{k+1}}( R^{(\beta)}) &= \frac{\partial^{k+1} J^{(\beta)}}{\partial r^{k+1}}(0) + \sum_{i=0}^d x_i\left(\frac{\partial^{k+2} J^{(\beta)}}{\partial r^{k+2}}(0)  - \frac{\partial^{k+1} I^{(\beta)}}{\partial r^{k+1}}(\alpha_i;0)\right)  + \cdots\\
	&= (-1)^{k} \frac{(\beta/2)^{2k}}{k!} - (-1)^{k+1} \frac{(\beta/2)^{2k+2}}{(k+1)!}\sum_{i=0}^d x_i\left(1 - \left(\frac{\alpha_i^2}{\beta^2}-1\right)^{k+1}\right) + \cdots.
\end{align*}
Plugging this into \eqref{eq:Mbeta} and performing the simple binomial sums we obtain the claimed expansion.
\end{proof}

\section{Weil--Petersson volumes and intersection numbers}\label{sec:wpintersection}

\subsection{Relation to generating function of intersection numbers}\label{sec:intersection}
Since the work of Kontsevich proving Witten's conjecture \cite{Kontsevich1992} it is well-known that in the absence of irreducibility constraints volumes of genus-$g$ metric maps are closely related to intersection theory on the moduli space of genus $g$ curves.
To be precise, let $\overline{\mathcal{M}}_{g,n}$ be the Deligne--Mumford compactification of the moduli space $\mathcal{M}_{g,n}$ of genus-$g$ curves with $n$ marked points.
It comes with $n$ tautological line bundles, arising naturally from the cotangent spaces at each of the $n$ marked points.
The Chern classes of these line bundles are denoted $\psi_1,\ldots,\psi_n$. 
The \emph{$\psi$-class intersection numbers} correspond to integrals of products of these over $\overline{\mathcal{M}}_{g,n}$,
\begin{equation*}
	\langle \tau_{d_1}\tau_{d_2}\cdots \tau_{d_n} \rangle_{g,n} \coloneqq \int_{\overline{\mathcal{M}}_{g,n}} \psi_1^{d_1}\psi_2^{d_2}\cdots \psi_n^{d_n}, \qquad d_1,\ldots,d_n \in \Z_{\geq 0}, \quad d_1+\cdots+d_n = 3g-3+n,
\end{equation*}
with the latter condition entering because $\overline{\mathcal{M}}_{g,n}$ has dimension $6g-6+2n$.
Kontsevich showed in \cite[Section 3]{Kontsevich1992} (see also \cite{Norbury2010}) that the volume of genus-$g$ metric maps with $n$ faces of circumference $\alpha_1,\ldots,\alpha_n$ is given by
\begin{equation*}
V_{g,n}^{(0)}(\alpha_1,\ldots,\alpha_n) = 2^{5-5g-2n} \sum_{\substack{d_1,\ldots,d_n\geq 0\\ d_1+\cdots+d_n = 3g-3+n}} \langle \tau_{d_1}\tau_{d_2}\cdots \tau_{d_n} \rangle_{g,n} \prod_{i=1}^n \frac{\alpha_i^{2d_i}}{\,d_i!}.
\end{equation*}
Besides the $\psi_1,\ldots,\psi_n$ there is another important class $\kappa_1$, the \emph{first Mumford--Morita--Miller class}, which is also the cohomology class of the Weil--Petersson symplectic form (up to a factor $2\pi^2$).
One may then consider mixed $\psi,\kappa_1$-intersection numbers
\begin{equation*}
	\langle \kappa_1^m\tau_{d_1}\cdots \tau_{d_n} \rangle_{g,n} \coloneqq \int_{\overline{\mathcal{M}}_{g,n}} \kappa_1^m\psi_1^{d_1}\cdots \psi_n^{d_n}, \qquad m,d_1,\ldots,d_n \in \Z_{\geq 0}, \quad m+d_1+\cdots+d_n = 3g-3+n,
\end{equation*}
Mirzakhani has shown \cite{Mirzakhani2007a} that the Weil--Petersson volumes of genus-$g$ hyperbolic surfaces with $n$ boundary components of lengths $L_1, \ldots, L_n$ are related to these via
\begin{equation*}
V_{g,n}^{\mathrm{WP}}(L_1,\ldots,L_n) = 2^{3-3g-n} \sum_{\substack{d_1,\ldots,d_n,m\geq 0\\ d_1+\cdots+d_n+m = 3g-3+n}} \langle \kappa_1^{m} \tau_{d_1}\tau_{d_2}\cdots \tau_{d_n} \rangle_{g,n}  \frac{(2\pi)^{2m}}{\,m!} \prod_{i=1}^n \frac{L_i^{2d_i}}{\,d_i!}.
\end{equation*}
Let us start by reformulating our partition functions $F_{g}^{(0)}$ and $F_g^{\mathrm{WP}}$ in terms of the conventional formal power series
\begin{align*}
G_g(s,t_0,t_1,\ldots) &= \sum_{n=0}^\infty\sum_{\substack{m,n_0,n_1,\ldots\geq 0\\ m+\sum_i i n_i = 3g-3+n\\
		\sum_in_i=n}} \langle \kappa_1^m\tau_0^{n_0}\tau_1^{n_1}\cdots\rangle_{g,n} \frac{s^m}{m!}\prod_{i\geq 0} \frac{t_i^{n_i}}{n_i!}. 
\end{align*}

\begin{lemma}\label{lem:FG} For $g\geq 0$ we have
	\begin{align*}
	F_g^{(0)}(x_0,\ldots,x_d) &= 2^{g-1} G_g(0, t_0, t_1,\ldots)\;\,\qquad\text{with}\qquad t_k = \frac{4^{-k}}{k!} \sum_{i=0}^d x_i\, \alpha_i^{2k},\\
	F_g^{\mathrm{WP}}(x_0,\ldots,x_d) &= 2^{g-1} G_g(\pi^2, t_0, t_1,\ldots) \qquad\text{with}\qquad  t_k = \frac{4^{-k}}{k!} \sum_{i=0}^d x_i\, L_i^{2k}.
	\end{align*}
\end{lemma}
\begin{proof}
	For the case of Weil--Petersson volumes we find
\begin{align*}
	F_g^{\mathrm{WP}}(x_1,\ldots,x_d) & =  \sum_{n=1}^\infty\frac{1}{n!} \sum_{i_1=0}^d x_{i_1} \cdots\! \sum_{i_n=0}^d x_{i_n} 2^{2-2g-n} V^{\mathrm{WP}}_{g,n}(L_{i_1}\ldots,L_{i_n})     \\
	                                  & =\sum_{n=1}^\infty\frac{1}{n!} \sum_{i_1=0}^d x_{i_1} \cdots\! \sum_{i_n=0}^d x_{i_n} 2^{5-5g-2n}\!\!\!\!\!\!\sum_{\substack{d_1,\ldots,d_n,m\geq 0 \\ d_1+\cdots+d_n+m = 3g-3+n}} \!\!\!\langle \kappa_1^{m} \tau_{d_1}\tau_{d_2}\cdots \tau_{d_n} \rangle_{g,n}  \frac{(2\pi)^{2m}}{\,m!} \prod_{k=1}^n \frac{L_{i_k}^{2d_k}}{\,d_k!}\\
	                                  & =2^{g-1}\sum_{n=1}^\infty\frac{1}{n!} \sum_{i_1=0}^d x_{i_1} \cdots\! \sum_{i_n=0}^d x_{i_n} \!\!\!\!\!\!\sum_{\substack{d_1,\ldots,d_n,m\geq 0     \\ d_1+\cdots+d_n+m = 3g-3+n}} \!\!\!\langle \kappa_1^{m} \tau_{d_1}\tau_{d_2}\cdots \tau_{d_n} \rangle_{g,n}  \frac{\pi^{2m}}{m!} \prod_{k=1}^n \frac{L_{i_k}^{2d_k}}{4^{d_k}\,d_k!}
\end{align*}
and by interchanging the order of summations
\begin{align*}
	F_g^{\mathrm{WP}}(x_1,\ldots,x_d) & =2^{g-1}\sum_{n=1}^\infty\frac{1}{n!} \!\!\!\sum_{\substack{d_1,\ldots,d_n,m\geq 0                                            \\ d_1+\cdots+d_n+m = 3g-3+n}} \!\!\!\langle \kappa_1^{m} \tau_{d_1}\tau_{d_2}\cdots \tau_{d_n} \rangle_{g,n}  \frac{\pi^{2m}}{m!} \prod_{k=1}^n \frac{4^{-d_k}}{d_k!}\sum_{i=0}^d x_{i}\,L_{i}^{2d_k}\\
	 & =2^{g-1}\sum_{n=1}^\infty\sum_{\substack{m,n_0,n_1,\ldots\geq 0                                                               \\ m+\sum_i i n_i = 3g-3+n\\
			\sum_in_i=n}} \langle \kappa_1^m\tau_0^{n_0}\tau_1^{n_1}\cdots\rangle_{g,n} \frac{\pi^{2m}}{m!}\prod_{j\geq 0} \frac{t_j^{n_j}}{n_j!}. \\
	 & = 2^{g-1} G_g(\pi^2, t_0,t_1,\ldots).
\end{align*}
The computation in the case of $F_g^{(0)}$ is entirely analogous.	
\end{proof}

\subsection{Substitution relations for $g\geq 2$}
Note that we have already established
\begin{equation*}
	F^{\mathrm{WP}}_0 = \frac{1}{4}\int_{0}^{R} \!\!\!\rmd r\, Z(r)^2,\qquad
	F^{\mathrm{WP}}_1 = - \frac{1}{24} \log M^{\mathrm{WP}}_0
\end{equation*}  
as a consequence of Proposition \ref{thm:metricmapgf} and Corollary \ref{thm:metricwpequivalence}.
Therefore we focus on the case $g\geq 2$.

Expression for the generating functions $G_g(s,t_0,t_1,\ldots)$ have been obtained in many places in the literature (although often for special cases like $s=0$), see \cite[Section 5]{Itzykson_Combinatorics_1992} for an early reference and \cite{Itzykson_Combinatorics_1992,Zograf_Weil_1998,Dubrovin_Normal_2001,Okuyama_JT_2020}.
The general idea is that one can relate $G_g(s,t_0,t_1,\ldots)$ to the corresponding generating function with its first one, two or three arguments set to zero via appropriate argument substitutions. 
Schematically,
\begin{equation*}
	G_g(s,t_0,t_1,\ldots) \xrightarrow{\kappa_1\text{ to }\psi\text{-class}}  G_g(0,t_0,t_1,\ldots) \xrightarrow{\text{string equation}} G_g(0,0,t_1,t_2,\ldots)\xrightarrow{\text{dilaton equation}} G_g(0,0,0,t_2,\ldots).
\end{equation*}

The first substitution relies on a relation between intersection numbers involving $\kappa_1$ and pure $\psi$-class intersection numbers \cite[Section 2]{Witten_Two_1991}.
It implies that \cite[Theorem 4.1]{Manin_Invertible_2000} (see also \cite{Kaufmann_Higher_1996,Mulase_Mirzakhanis_2008,Liu_Recursion_2009,Bertola_Correlation_2016})
\begin{equation}\label{eq:wpsubstitution}
G_g(s,t_0,t_1,t_2,t_3\ldots) = G_g(0,t_0,t_1,t_2+\gamma_2,t_3+\gamma_3, \ldots), \qquad \gamma_k = \frac{(-1)^k}{(k-1)!} s^{k-1} \,\ind_{k\geq 2}.
\end{equation}

We will include derivations of the other two substitutions, based on the approach of \cite{Itzykson_Combinatorics_1992,Itzykson_Combinatorics_1992}, since they admit short proofs that are not quite spelled out in the literature (in the required form).
The starting points are the \emph{string equation} \cite[(2.22)]{Witten_Two_1991}
\begin{equation}
	\frac{\partial G_g}{\partial t_0}(0,t_0,t_1,\ldots) = \sum_{i=0}^\infty t_{i+1} \frac{\partial G_g}{\partial t_i}(0,t_0,t_1,\ldots)
\end{equation}
and \emph{dilaton equation} \cite[(2.47)]{Witten_Two_1991}
\begin{equation}\label{eq:Gdilaton}
	\frac{\partial G_g}{\partial t_1}(0,t_0,t_1,\ldots) = \left(2g-2+\sum_{i=0}^\infty t_i \frac{\partial }{\partial t_i}\right)G_g(0,t_0,t_1,\ldots)
\end{equation}
satisfied by the generating function of $\psi$-class intersection numbers.
Note that these are closely related to (or rather special cases of) the generalized string and dilaton equations \eqref{eq:stringwp} and \eqref{eq:dilatonwp}.
\begin{lemma}
For genus $g\geq 2$ we have
\begin{align}
		G_g(0,t_0,t_1,t_2,\ldots) &= G_g(0,0,f_1,f_2,\ldots),\label{eq:stringcons}\\
		G_g(0,0,t_1,t_2,\ldots) &= (1-t_1)^{2-2g} G_g\left(0,0,0,\frac{t_2}{1-t_1},\frac{t_3}{1-t_1},\ldots\right),\label{eq:dilatoncons}
\end{align}
where $f_p(t_0,t_1,\ldots) = \sum_{k=0}^\infty \frac{t_{k+p}}{k!} u^k$ and $u(t_0,t_1,\ldots)$ is the formal power series solution to $u = \sum_{k=0}^\infty \frac{t_k}{k!} u^k$.
\end{lemma}
\begin{proof}
By taking a $t_0$-derivative of $u = \sum_{k=0}^\infty \frac{t_k}{k!} u^k$ we observe that $f_1 = 1 - 1 / \partial_{t_0} u$ and therefore for $p\geq 1$ we have
\begin{equation*}
	f_{p+1} = (1-f_1)\,\partial_{t_0} f_p.
\end{equation*}
Note that the string equation can be rewritten as
\begin{align*}
	\frac{\partial G_g}{\partial t_0} = \frac{1}{1-t_1} \sum_{i=1}^\infty t_{i+1} \frac{\partial G_g}{\partial t_i}.
\end{align*}
Using these identities we observe that
\begin{align*}
	\frac{\partial }{\partial s}G_g(0,t_0-s,f_1(s,t_1,\ldots),f_2(s,t_1,\ldots),\ldots) &= -\frac{\partial G_g}{\partial t_0} + \sum_{p=1}^\infty\frac{\partial f_p}{\partial t_0}  \frac{\partial G_g}{\partial t_{p}}\\
	&= -\frac{\partial G_g}{\partial t_0} - \frac{1}{1-f_1}\sum_{p=0}^\infty f_{p+1}\frac{\partial G_g}{\partial t_{p}}=0.
\end{align*}
Integrating the left-hand side from $s=0$ to $s = t_0$ we obtain the identity \eqref{eq:stringcons}.

Similarly the dilaton equation implies that
\begin{align*}
	\frac{\partial}{\partial s}\left( s^{2g-2} G_g(0,0,1-s,s t_2,s t_3, \ldots)\right) &= - s^{2g-2}\frac{\partial G_g}{\partial t_1}+(2g-2) s^{2g-3} G_g  + s^{2g-2}\sum_{i=2}^\infty t_i \frac{\partial G_g}{\partial t_i} \\
	&= s^{2g-3} \left( -\frac{\partial G_g}{\partial t_1} + (2g-2) G_g+ (1-s)\frac{\partial G_g}{\partial t_1}+ \sum_{i=2}^\infty st_i \frac{\partial G_g}{\partial t_i}\right) \stackrel{\eqref{eq:Gdilaton}}{=} 0.
\end{align*}
Integrating the left-hand side from $s=1-t_1$ to $s=1$ we obtain the relation \eqref{eq:dilatoncons}.
\end{proof}

\subsection{Proof of Theorem~\ref{thm:genfun}}\label{sec:proofgenfun}

Combining Lemma \ref{lem:FG} with \eqref{eq:wpsubstitution} we see that
\begin{align*}
	&F_g^{\textrm{WP}}(x_0,\ldots,x_d) = 2^{g-1}G_g(0,0,f_1,f_2,\ldots),\\
	&f_p = \sum_{k=0}^\infty \frac{t_{k+p}+\gamma_{k+p}}{k!} u^k, \qquad t_k=\frac{2^{-2k}}{k!} \sum_{i=0}^d x_i\, L_i^{2k}, \qquad \gamma_k =\frac{(-1)^k}{(k-1)!} \pi^{2k-2} \,\ind_{k\geq 2},
\end{align*}
where $u$ is the formal power series solution to
\begin{align*}
	0 = u - \sum_{k=0}^\infty \frac{t_k + \gamma_k}{k!} u^k = \frac{\sqrt{u}}{\pi}J_1(2\pi\sqrt{u}) - \sum_{i=0}^d x_i \,I_0(L_i\sqrt{u}) \stackrel{\eqref{eq:Rdef}}{=} Z(u).
\end{align*}
This means that $u$ expressed in terms of $x_0,\ldots,x_d$ is nothing but our formal power series $R$.
Moreover, comparing $f_1 = 1 - 1 / \partial_{t_0} u$ and $f_{p+1} = (1-f_1)\,\partial_{t_0} f_p$ (see the proof of the lemma above) to \eqref{eq:Mrecurrence}, we conclude that
\begin{equation*}
	f_k = \ind_{\{k=1\}} - M_{k-1}^{\mathrm{WP}}.
\end{equation*}
Applying the further substitution \eqref{eq:dilatoncons} we thus arrive at 
\begin{align*}
	F_g^{\textrm{WP}}(x_0,\ldots,x_d) &= 2^{g-1} (M_0^{\mathrm{WP}})^{2-2g}G_g\left(0,0,0,\frac{-M_1^{\mathrm{WP}}}{M_0^{\mathrm{WP}}},\frac{-M_2^{\mathrm{WP}}}{M_0^{\mathrm{WP}}},\ldots\right) \\
	&\stackrel{\eqref{eq:Ppol}}{=}2^{g-1} (M_0^{\mathrm{WP}})^{2-2g}\,\mathcal{P}_g\left(\frac{M_1^{\mathrm{WP}}}{M_0^{\mathrm{WP}}},\frac{M_2^{\mathrm{WP}}}{M_0^{\mathrm{WP}}},\ldots\right).
\end{align*}
This proves the formula for $F_g^{\mathrm{WP}}$ of Theorem~\ref{thm:genfun}.

The analogous computation starting from $F_g^{(0)}$ of Lemma \ref{lem:FG} yields 
\begin{align*}
	F_g^{(0)}(x_0,\ldots,x_d) &= 2^{g-1} (M_0^{(0)})^{2-2g}\,\mathcal{P}_g\left(\frac{M_1^{(0)}}{M_0^{(0)}},\frac{M_2^{(0)}}{M_0^{(0)}},\ldots\right),
\end{align*}
while Proposition \ref{thm:metricmapgf} implies
\begin{align*}
	F_g^{(0)}(x_0,\ldots,x_d) &= 2^{g-1} (M_0^{(0)})^{2-2g}\,\bar{\mathcal{P}}_g\left(\frac{M_1^{(0)}}{M_0^{(0)}},\frac{M_2^{(0)}}{M_0^{(0)}},\ldots\right).
\end{align*}
For proper choice of $d$ and $\alpha_1,\ldots,\alpha_d$ the formal power series $M_p^{(0)}$, $p\geq 0$, are algebraically independent, so we may conclude that the polynomials agree, $\bar{\mathcal{P}}_g = \mathcal{P}_g$ for every $g\geq 2$.
This finishes the proof of Theorem~\ref{thm:genfun}.

\appendix
\section{Proof of Proposition \ref{thm:discretepolyhedra}}
According to Rivin's theorem, ideal convex polyhedra with $n$ labeled vertices and dihedral angles in $\frac{\pi}{b}\Z$ are in bijection with (unrooted) $2\pi$-irreducible planar maps with $n$ labeled faces of circumference $2\pi$ and edges of length in $\frac{\pi}{b}\Z$.
After splitting each edge of length $\frac{\pi}{b}k$ into $k$ consecutive edges and forgetting the length assignments, one obtains precisely the set of (unrooted) $2b$-irreducible planar maps with $n$ labeled faces of degree $2b$.
These were enumerated in \cite[Section 9.1]{Bouttier2014}.
According to \cite[Theorem 1]{Budd2020} one may express
\begin{equation*}
	P_n^{(b)} = \hat{N}_{0,n}^{(b)}(b,b\ldots) + \frac{(n-1)!}{2}(-1)^n,
\end{equation*}
which is indeed polynomial in $b$ of degree $2n-6$.
By \cite[Proposition 15]{Budd2020} we have
\begin{align*}
	\hat{N}_{0,n}^{(b)}(b,b\ldots) = (n-2)! [z^{n-2}] \int_0^{J^{-1}(b;z)} (1+r)^{-2b-1} = -\frac{1}{2b}(n-2)! [z^{n-2}](1+J^{-1}(b;z))^{-2b},
\end{align*}
where we used that $I(b,b;r) = 1$. 
Together with the expansion 
\begin{equation*}
	-\frac{1}{2b}(1+J^{-1}(b;z))^{-2b} = -\frac{1}{2b} + z + \cdots
\end{equation*}
this implies that
\begin{equation*}
	\sum_{n=4}^\infty P_n^{(b)} \frac{z^{n-2}}{(n-2)!} = \frac{1}{2(1+z)^2}- \frac{1}{2} + \frac{1}{2b}-\frac{1}{2b} \left(1+J^{-1}(b;z) \right)^{-2b}.
\end{equation*}
Note that an expression like this was already pretty much contained in \cite[(9.21)]{Bouttier2014} (see also \cite[(2.7)]{Bouttier_Planar_2019}).

\end{document}